\documentclass{article}
\usepackage[utf8]{inputenc}
\usepackage[round]{natbib}
\RequirePackage[OT1]{fontenc}
\RequirePackage{amsthm,amsmath}
\allowdisplaybreaks[4]
\usepackage{amssymb}
\usepackage{fullpage}	
\usepackage{amsthm}
\usepackage{diagbox}
\usepackage{enumerate}
\usepackage{bbm}
\usepackage{bm}
\usepackage{float}
\usepackage{graphicx}
\usepackage[ruled,linesnumbered]{algorithm2e}
\usepackage{xr-hyper}
\usepackage{hyperref}[]
\hypersetup{
	colorlinks=true,
	linkcolor=blue,
	filecolor=magenta,      
	urlcolor=cyan,
	citecolor=blue,
}
\usepackage{cleveref}
\usepackage{subcaption}
\usepackage[dvipsnames]{xcolor}
\usepackage{multirow}
\usepackage[hmargin={1.0in, 1.0in}, vmargin={1.2in, 1.2in}]{geometry}
\usepackage{booktabs}
\usepackage{array}
\usepackage{enumitem,epsfig,natbib}
\usepackage{multirow}
\usepackage{multicol}
\newcolumntype{Z}{>{\setbox0=\hbox\bgroup}c<{\egroup}@{\hspace*{-\tabcolsep}}}
\usepackage[title]{appendix}

\newtheorem{thm}{Theorem}[section]
\newtheorem{ass}{Assumption}[section]

\newtheorem{cor}{Corollary}[section]
\newtheorem{pro}{Proposition}[section]
\newtheorem{rem}{Remark}[section]

\newtheorem{alg}{Algorithm}[section]
\newtheorem{defn}{Definition}[section]
\newtheorem{lem}{Lemma}[section]

\usepackage{setspace}
\usepackage{diagbox}
\usepackage{authblk}
\usepackage{xparse}
\usepackage{tikz}
\usetikzlibrary{positioning}
\usetikzlibrary{plotmarks}
\usetikzlibrary{matrix}
\usepackage{pgfplots}
\pgfplotsset{compat=1.7}
\usetikzlibrary{matrix,backgrounds, arrows.meta}
\pgfdeclarelayer{myback}
\pgfsetlayers{myback,background,main}
\usetikzlibrary{decorations.pathreplacing,angles,quotes}
\tikzset{mycolor/.style = {dashed,rounded corners,line width=1bp,color=#1}}%
\tikzset{myfillcolor/.style = {draw,fill=#1}}%
\tikzset{
	declare function={
		normcdf(\x,\m,\s)=1/(1 + exp(-0.07056*((\x-\m)/\s)^3 - 1.5976*(\x-\m)/\s));
	}
}

\newcommand{\cb}{\mathcal{B}}

\newcommand{\ck}{\mathcal{K}}

\newcommand{\cn}{\mathcal{N}}

\newcommand{\cs}{\mathcal{S}}

\newcommand{\cw}{\mathcal{W}}

\newcommand{\br}{\mathbb{R}}

\newcommand{\chol}{\mbox{chol}}
\newcommand{\logm}{\mbox{logm}}

\newcommand{\stsim}[1]{\stackrel{#1}{\sim}}

\newcommand{\bigp}[1]{\big({#1}\big)}
\newcommand{\lrp}[1]{\left({#1}\right)}

\newcommand{\lrabs}[1]{\left|{#1}\right|}

\newcommand{\lrfl}[1]{\lfloor{#1}\rfloor}

\newcommand{\cvec}[2]{
  \left(
    \begin{array}{c}
      #1 \\
      #2 \\
    \end{array}
  \right)}

\newcommand{\mat}[4]{
  \left(
    \begin{array}{cc}
      #1 & #2 \\
      #3 & #4 \\
    \end{array}
  \right)}

\newcommand{\BE}{\begin{equation}}
\newcommand{\EE}{\end{equation}}
\newcommand{\BEqn}{\begin{eqnarray*}}
\newcommand{\EEqn}{\end{eqnarray*}}

\makeatletter
\newcommand*{\addFileDependency}[1]{
  \typeout{(#1)}
  \@addtofilelist{#1}
  \IfFileExists{#1}{}{\typeout{No file #1.}}
}
\newcommand*\showfontsize{\f@size{} point}
\makeatother

\newcommand{\blind}{1}

\begin{document}
 
\if1\blind
{
 	\title{Testing Serial Independence of Object-Valued Time Series}	
		
     \author[1]{Feiyu Jiang}
     \author[2]{Hanjia Gao}
     \author[2]{Xiaofeng Shao}
        \affil[1]{Department of Statistics and Data Science, Fudan University}

    \affil[2]{Department of Statistics, University of Illinois at Urbana Champaign}
		
	\date{}	
	\maketitle
} \fi


\begin{abstract}
We propose a novel method for testing serial independence of object-valued time series in metric spaces, which is more general than Euclidean or Hilbert spaces. The proposed method is fully nonparametric, free of tuning parameters and can capture all nonlinear pairwise dependence.  The key concept used in this paper is the distance covariance in metric spaces, which is extended to auto-distance covariance for object-valued time series.   
Furthermore, we propose a generalized spectral density  function to account for pairwise dependence at all lags and construct  a Cram\'er von–Mises type test statistic. New theoretical arguments are developed to establish the asymptotic behavior of the test statistic. A wild bootstrap is also introduced to obtain the critical values of the non-pivotal limiting null distribution.  Extensive numerical simulations and {two}  real data applications are conducted to illustrate the effectiveness and versatility of our proposed test.
\end{abstract}



\section{Introduction}
Random objects in general metric spaces have become increasingly common in modern statistical and econometric research. For example, the intraday return path of a financial asset \citep{aue2017functional},  the annual composition of energy sources \citep{zhu2023spherical}, social networks \citep{board2021learning}, EEG scans or MRI fiber tracts of  patients \citep{kurtek2012statistical} can all be viewed as random objects in certain metric spaces, although they are typically given specific names  as functional data, compositional data, network data, image data and curve data, among others.  The concept of random objects also undoubtedly include classical notions such as vectors, covariance matrices and distributions.  Instead of building specific models for each of them, by viewing these objects as random elements in metric spaces, we may be able to simplify the modeling and inference while preserving the ability of extracting meaningful information and patterns in a unified fashion \citep{petersen2019frechet, dubey2019frechet, dubeymuller2020, zhangJTSA2022}. 


For many endeavors in this area, the data they analyzed is collected with a natural ordering, i.e., the data is object-valued time series.   However, most existing modeling and inference techniques either presume temporal independence or construct time series models without conducting diagnostic checking to assess the goodness-of-fit. This is mainly due to the unavailability of appropriate tests. As a result, researchers often overlook the serial dependence in their data and fail to account for its impact on their analyses. Consequently, the validity and reliability of their findings could be compromised. This motivates us to develop a new test for serial independence of object-valued time series. 


 Testing serial independence has a long and rich history in statistics and econometrics, with a vast literature that cannot be exhaustively listed. The early work  dates back to \cite{box1970distribution} and \cite{ljung1978measure}. Since then, numerous tests have been proposed for univariate random variables 
 or multivariate vectors in Euclidean space, with much attention devoted to testing for second-order uncorrelatedness,   see  \cite{li1981distribution}, \cite{deo2000spectral}, \cite{lobato2001testing}, \cite{escanciano2009automatic},  \cite{shao2011bootstrap}, to name a few. These tests typically capture linear serial dependence in data, and have no power against nonlinear dependence. Nonlinear serial dependence is indeed  prevalent among many real-world time series, and many parametric nonlinear models have been proposed to capture nonlinear dependence. A prominent example is GARCH model, which  implies uncorrelatedness but is serially dependent. 
 Note that several tests have been developed to target at  higher order dependence \citep{li1994squared,ling1997diagnostic} and  at general nonlinear dependence \citep{hong1999hypothesis,escanciano2006generalized}.   We also note the recent developments for testing white noise hypothesis in  functional time series in Hilbert space, see e.g.  \cite{gabrys2007portmanteau}, \cite{horvath2013test} and \cite{zhang2016white}.

Despite many tests available for testing the serial independence/uncorrelatedness  in Euclidean and Hilbert spaces, they cannot be directly used for  testing serial independence of  object-valued time series, because of the lack of classic  algebraic operations in general metric space, such as addition,  multiplication and taking inner product.
To fill this gap, we propose to build a new test based on 
\textit{distance covariance}, which was originally proposed by \cite{szekely2007measuring} to measure dependence among two random vectors in Euclidean space using characteristic functions, and later extended  by \cite{zhou2012measuring} into time series setting by using \textit{auto-distance covariance} (ADCV). A crucial feature of the  (auto-)distance covariance is that it can capture both linear and nonlinear serial dependence in data. Although many researchers have employed this idea into serial dependence testing, they are only valid  in Euclidean space \citep{fokianos2017consistent, fokianos2018testing, davis2018applications}.   

Based on the concept of \textit{distance covariance in metric space} by \cite{lyons2013distance}, we are able to extend ADCV in Euclidean space into general metric spaces, which then naturally serves the goal of testing independence at fixed lags. To take into account   pairwise  (in)dependence at all lags, we then propose  a generalized spectral density/distribution function using ADCVs in the same spirit of classical  spectral density/distribution function using autocovariances. Following the developments in \cite{shao2011bootstrap}, who proposed 
a spectrum-based test for white noise hypothesis of a univariate time series, we develop a Cram\'er von–Mises (CvM) type test statistic and study the limiting behavior of the test under both the null and alternatives. { Our test significantly enhances classical spectrum-based tests in two fundamental ways. First, it captures both linear and nonlinear serial dependencies, a capability lacking in conventional methods except for \cite{hong1999hypothesis} and \cite{escanciano2006generalized}. Second, its versatility extends to data objects in general metric spaces, which is much broader than the Euclidean or Hilbert space considered in the literature. To the best of our knowledge, this is the first formal attempt at testing temporal independence for object-valued time series in a unified manner.} Unlike conventional autocovariance function or spectrum based tests for Euclidean time series, the estimation of ADCVs is based on U-statistics. Therefore, new theoretical arguments are developed to establish the asymptotic theory for the proposed test statistic. Since the limiting null distribution is nonstandard and non-pivotal, we propose a wild bootstrap approach to  facilitate  practical implementation of our test. The bootstrap consistency is also established.

We now introduce the notation. Let $(\Omega,d)$ be a separable metric space, and let $(\mathbb{S},\mathbb{P},\mathcal{F})$ be the probability space. For a vector $a\in\mathbb{R}^q$, we denote its Euclidean norm as $|a|_q$. Denote $``\to_p"$ and  $``\to_d"$ the convergence in probability and in distribution, respectively. Denote $L_2[0,\pi]$ the Hilbert space $\mathbb{H}$ of all square integrable functions on $[0,\pi]$ (with respect to Lebesgue measure) with inner product $\langle f,g\rangle=\int_{[0,\pi]}f(\zeta)g^c(\zeta)\mathrm{d}(\zeta)$ where $g^c(\zeta)$ denotes the complex conjugate of $g(\zeta)$, and the norm $\|f\|=\langle f,f\rangle^{1/2}$. We denote $``\Rightarrow"$ as weak convergence in $L_2[0,\pi]$. 

The rest of the paper is organized as follows. We first provide  backgrounds of (auto-) distance covariance in metric spaces in \Cref{Sec:pre}.  \Cref{Sec:test} then introduces our distance covariance based test statistics, and investigates their asymptotic distributions under null and alternatives. \Cref{Sec:boot} provides the wild bootstrap algorithm for approximating the limiting null distribution. Extensive numerical experiments are conducted in \Cref{Sec:simu} with competing methods for testing serial independence of {the functional time series in Hilbert space, the covariance matrix time series, and the univariate distributional time series.}  \Cref{Sec:app} illustrates the usefulness and versatility of our tests via {two} meaningful real data applications in financial data and human mortality data.
\Cref{Sec:con} concludes. Additional numerical results and all the technical proofs are provided in the supplement. {The code is available at \url{https://github.com/hjgao117/JiangGaoShao}.}


\section{Preliminaries}\label{Sec:pre}
In this section, we  provide some background  on  the  concept of distance covariance in metric spaces and its use in quantifying dependence.  The  extension to auto-distance covariance (ADCV) for the  time series setting is also introduced. 
\subsection{Distance covariance}
The concept of \textit{distance covariance} was first introduced by \cite{szekely2007measuring} as a measure of dependence between two random vectors $X\in\mathbb{R}^p$ and $Y\in\mathbb{R}^q$. It is defined as the weighted integral of the discrepancy between the joint characteristic function and the product of the marginal characteristic functions of $(X,Y)$.
\begin{defn}[Distance Covariance in Euclidean Space]\label{dfn_dcovE}
    For $X\in\mathbb{R}^p$ and $Y\in\mathbb{R}^q$, the distance covariance is given  by \begin{flalign*}
   &\mathrm{dcov}(X,Y)\\=&\int_{\mathbb{R}^{p+q}} |\mathbb{E}\{\exp[i(t'X+s'Y)]\}-\mathbb{E}\{\exp[it'X]\}\mathbb{E}\{\exp[is'Y]\}|^2 (c_pc_q|t|_p^{1+p}|s|_q^{1+q})^{-1}\mathrm{d}t\mathrm{d}s 
\end{flalign*}
where $c_d=\pi^{(1+d)/2}/\Gamma((1+d)/2)$ and $\Gamma$ is the Gamma function. 
\end{defn}

Note that for  notational simplicity, we shall use $\mathrm{dcov}$ instead of $\mathrm{dcov}^2$ as used in the original definition of distance covariance. 
\cite{szekely2007measuring} also managed to derive the following alternative definition when $\mathbb{E}[|X|_p^2+|Y|_q^2]<\infty$, 
\begin{flalign}
    \notag& \mathrm{dcov}(X,Y)\\=&\label{dcovE2}\mathbb{E}\left|X-X^{\prime}\right|_p\left|Y-Y^{\prime}\right|_q+\mathbb{E}\left|X-X^{\prime}\right|_p \mathbb{E}\left|Y-Y^{\prime \prime}\right|_q-2 \mathbb{E}\left|X-X^{\prime}\right|_p\left|Y-Y^{\prime \prime}\right|_q,
\end{flalign}
with $(X',Y')$ and $(X'',Y'')$ being independent copies of $(X,Y)$. 

Distance covariance can be used to characterize the dependence between $X$ and $Y$ due to the following crucial property. \begin{pro}
\label{property}
$\mathrm{dcov}(X,Y)\geq 0$, and the equality holds iff $X$ is independent of $Y$. 
\end{pro}

\Cref{property} has proven to be a very powerful tool  in Euclidean space, with its use widely appeared in mutual dependence testing \citep{yzs2018}, feature screening \citep{lzz2013}, dimension reduction \citep{shengyin2016}, among many other applications. However, \Cref{dfn_dcovE} is not easily extended to general metric spaces  because conventional algebraic manipulation such as addition and  multiplication  may not be applied.  
Instead, by viewing $|\cdot|_q$ as a metric in Euclidean space, the alternative definition in \eqref{dcovE2} is extendable.
\subsection{Distance covariance in metric space}
Let $(\Omega,d)$ be a separable metric space. Let $M(\Omega)$ be the set of probability measures on $\Omega$, we denote the subset  of $M(\Omega)$ possessing $p$th moment as $$M_p(\Omega)=\{\nu\in M(\Omega): \text{  for some }\omega\in\Omega, 
\int_{\Omega} d^p(\omega,x) \mathrm{d}\nu(x)<\infty\}.
$$
For $\nu\in M_1(\Omega)$, define \citep[Lemma ~2.1]{lyons2013distance}  {\begin{flalign}\label{dfn_d}
\begin{split}
    &d^{(1)}(x)=\int_{\Omega} d(x,x')\mathrm{d}\nu(x'),\quad D=\int_{\Omega}  d^{(1)}(x)\mathrm{d}\nu(x), \\ &d_{\nu}(x,x')= d(x,x')-d^{(1)}(x)-d^{(1)}(x')+D.
\end{split}
\end{flalign}
}

Clearly, for two independent random objects $X,X'\sim\nu(\cdot)$, we can write $d^{(1)}(X)=\mathbb{E}[d(X,X')|X]$, and $D=\mathbb{E}[d(X,X')]$. 

{Note that   if $X$ and $Y$ are two random vectors taking values in conventional $q$-dimensional Euclidean space (i.e., $(\Omega,d)=(\mathbb{R}^q,|\cdot|_q)$) with marginal distributions being $\nu_X$ and $\nu_Y$ respectively,   then   \eqref{dcovE2} is equivalent to $\mathrm{dcov}(X,Y)=\mathbb{E}[d_{\nu_X}(X,X')d_{\nu_Y}(Y,Y')]$. }  Therefore, \cite{lyons2013distance} proposed the following definition of distance covariance in general metric spaces. 
\begin{defn}[Distance Covariance in Metric Space]\label{dfn_dcov}
    For $X,Y$ taking values in $(\Omega,d)$  whose marginals are  $\nu_X,\nu_Y\in M_1(\Omega)$ respectively,   the distance covariance between $X$ and $Y$ is given as  
\begin{equation}\label{dcov}
   \mathrm{dcov}(X,Y)=\mathbb{E}[d_{\nu_X}(X,X')d_{\nu_Y}(Y,Y')], 
\end{equation}
where $(X',Y')$ is an independent copy of $(X,Y)$. 
\end{defn}

However, as pointed out by \cite{lyons2013distance}, in general metric spaces, the above definition alone is insufficient for  \Cref{property} to hold, and  additional  topological assumption is required. He then introduced  the concept of {\it strong negative type} to resolve this issue.  

{\begin{defn}[Strong Negative Type]\label{dfn_snt}
   We say $(\Omega,d)$ is of strong negative type, if 
for $\nu_1,\nu_2\in M_1(\Omega)$  such that $\nu_1\neq \nu_2$, and $\nu_-=\nu_1-\nu_2$, we have
   $$
   \iint_{\Omega^2} d(x_1,x_2) \mathrm{d}\nu_{-} (x_1)\mathrm{d}\nu_{-} (x_2)< 0.
   $$
\end{defn}}
 Note that the class of metric spaces of strong negative type is actually quite large, for example, every separable Hilbert space is of strong negative type. {However, we also note there are many spaces that do not satisfy \Cref{dfn_snt}, e.g., $\mathbb{R}^q$ with $L_p$-metric  for $3\leq q\leq \infty$, $2<p\leq \infty$ are not of strong negative type. }We refer to \cite{lyons2013distance,lyons2014hyperbolic,lyons2020strong} for more discussions and examples. With the notion  of strong negative type, distance covariance then completely characterizes the (in)dependence in metric space, given in the following proposition.
\begin{pro}\label{pro_dcov}
   [Theorem 3.11 in \cite{lyons2013distance}] If $(\Omega,d)$ is of  strong negative type,  for $X,Y$ taking values in $(\Omega,d)$  whose marginals are  $\nu_X,\nu_Y\in M_1(\Omega)$ respectively, \Cref{property} continues to hold. 
\end{pro}

\subsection{Auto-distance covariance}

{
The distance covariance, as defined in \Cref{dfn_dcov}, measures the dependence between two random objects and might not be readily applicable to the time series context. Therefore, in order to address this limitation, we introduce the concept of auto-distance covariance (ADCV). 
} This notion was first proposed by 
  \cite{zhou2012measuring} for conventional Euclidean valued time series {by measuring temporal dependence between $\{X_t\}_{t\in\mathbb{Z}}\in\mathbb{R}^q$ and its lagged observation $\{X_{t-k}\}_{t\in\mathbb{Z}}$ for a fixed lag order $k\in\mathbb{Z}$.} Here we generalize the idea  for time series objects in metric space. 
\begin{defn}[ADCV]\label{dfn_adcv}
    Assume that $\{X_t\}_{t\in\mathbb{Z}}$ is a sequence of strict stationary time series taking values in $(\Omega,d)$. For $k\in\mathbb{Z}$,  we call 
$$
V(k)=  \mathbb{E}[d_{\nu}(X_t,X'_{t})d_{\nu}(X_{t-k},X'_{t-k})],
$$
the auto-distance covariance of $X_t$ at lag $k$, where $\{X_t'\}_{t\in\mathbb{Z}}$ is an independent copy of $\{X_t\}_{t\in\mathbb{Z}}$.
\end{defn}

It is clear that $V(k)=V(-k)$ for $k<0$, and  $V(0)=\mathbb{E}[d_{\nu}^2(X_t,X_{t}')]>0$. In addition, by  \Cref{pro_dcov},   $V(k)=0$, $k\neq 0$, iff $X_t$ and $X_{t-k}$ are independent of each other. We will exploit this property to build our test statistics given in the next section.
\begin{rem}
   In \cite{zhou2012measuring}, the  \textsc{ADCV} for Euclidean valued time series is  based on \Cref{dfn_dcovE} by replacing $(X,Y)$ with $(X_t,X_{t-k})$.  However,  for general metric spaces, such treatment is no longer valid.  We thus call for \Cref{dfn_dcov}. It can be shown that for $(\Omega,d)=(\mathbb{R}^q,|\cdot|_q)$, \Cref{dfn_adcv} is equivalent to the one in \cite{zhou2012measuring}.
\end{rem}


\section{Test statistics and asymptotics}\label{Sec:test}
Given  a sequence of stationary random objects $\{X_t\}_{t=1}^n$ that reside in  a separable metric space $(\Omega,d)$ of strong negative type, we are interested in testing the serial independence of $\{X_t\}_{t=1}^n$.   {The hypothesis testing problem is formulated as  
$$
 H_0: \{X_t\}_{t=1}^n \quad\mbox{is i.i.d.}~~\mbox{v.s.}\quad  H_a: V(k)\neq  0 \quad \mbox{ for some  } k\neq 0.
$$}
\subsection{Test statistics at fixed lags}
To illustrate the idea, we first consider a relatively simpler task by forming a test based on ADCV at fixed lags $1\leq k\leq K$. This suggests that we  find an empirical estimator for $V(k)$. Intuitively, under $H_0$, $(X_{t'},X_{t'-k})$ naturally forms an independent copy of $(X_t,X_{t-k})$ if $t'\neq t$, which motivates us to estimate $V(k)$ by replacing terms involving (conditional) expectations in $d_{\nu}(X_t,X_{t-k})$ with their empirical counterparts. 

Motivated by the estimator in Euclidean space \citep{szekely2007measuring}, we  construct the empirical estimator  $V_n(k)$ by adopting the $\mathcal{U}$-centering approach in \cite{szekely2014partial}.
 Specifically, let  $a_{ij}(k)=d(X_i,X_j)$ for $k+1\leq i,j\leq n$,  we denote $\{\widetilde{a}_{i j}(k)\}_{i,j=k+1,\cdots,n}$ as its $\mathcal{U}$-centered version:
\begin{flalign}\label{atilde}
\widetilde{a}_{ij}(k)= \begin{cases}a_{ij} - \frac{\sum_{t=k+1}^{n} a_{it}}{n-k-2} -\frac{\sum_{t'=k+1}^{n} a_{t'j}}{n-k-2} +\frac{ \sum_{k+1\leq t\neq t'\leq n} a_{tt'}}{(n-k-1)(n-k-2)}, & i \neq j \\ 0, & i=j.\end{cases}
\end{flalign}
Define $\tilde{b}_{ij}(k)$ similarly for $b_{ij}(k)=d(X_{i-k},X_{j-k})$ with $k+1\leq i,j\leq n$. Intuitively, for large $n$, $\tilde{a}_{ij}(k)$ (or  $\tilde{b}_{ij}(k)$) approximates the value of $d_{\nu}(X_i,X_j)$ {(or $d_{\nu}(X_{i-k},X_{j-k})$). }

We then estimate $V(k)$ by 
\begin{equation*}
   V_n(k)= \frac{1}{(n-k)(n-k-3)}\sum_{i,j=k+1}^n\tilde{a}_{ij}(k)\tilde{b}_{ij}(k). 
\end{equation*}

    { The $\mathcal{U}$-centering approach  was originally proposed by \cite{szekely2014partial} to provide an   unbiased estimator of distance covariance. It is
adopted here  for simplifying technical analysis as we can alternatively rewrite $V_n(k)$ by the following fourth order U-statistic \citep{zhang2018conditional}, 
\begin{equation}\label{Vn_U}
V_n(k) = { n-k \choose 4}^{-1} \sum_{k+1\leq i<j<q<r\leq n} h(Z_i^{(k)},Z_j^{(k)},Z_q^{(k)},Z_r^{(k)}),
\end{equation}
where $Z_i^{(k)}=(X_i,X_{i-k})$, and the kernel function is given by 
\begin{flalign}\label{kernel}
h\left(Z_i,Z_j,Z_q,Z_r\right)=\frac{1}{24} \sum_{(i_1,i_2,i_3,i_4)}^{(i, j, q, r)}d(X_{i_1},X_{i_2})\left[d(Y_{i_3},Y_{i_4})+d(Y_{i_1},Y_{i_2})-2d(Y_{i_1},Y_{i_3})\right],
\end{flalign}
such that  $Z_{\ell}=(X_{\ell},Y_{\ell})$, $\ell\in\{i,j,q,r\}$, and $\sum_{(i_1,i_2,i_3,i_4)}^{(i,j,q,r)}$ denotes the summation over all 24($=4!$) permutations of the 4-tuple of indices $(i,j,q,r)$. 

\begin{rem}
Alternatively, one could use the conventional empirical  estimator of $V(k)$ given by 
$ \widehat{V}_n(k)= \frac{1}{(n-k)^2}\sum_{i,j=k+1}^n\hat{a}_{ij}(k)\hat{b}_{ij}(k), $
where $\hat{a}_{ij}(k)= a_{ij}(k)-\bar{a}_{i .}(k)-\bar{a}_{. j}(k)+\bar{a}_{. .}(k)$,    $$
\bar{a}_{i .}(k)=\frac{1}{n-k} \sum_{t=1}^n a_{i t}, \quad \bar{a}_{. j}(k)=\frac{1}{n-k} \sum_{t=k+1}^n a_{t j}, \quad \bar{a}_{. .}(k)=\frac{1}{(n-k)^2} \sum_{t,t'=k+1}^n a_{t,t'},
$$
with $\hat{b}_{ij}(k)$ similarly defined. Note that $ \widehat{V}_n(k)$ is a combination of $V$-statistics, the technical analysis could be  more involved than the $\mathcal{U}$-centering approach \citep{szekely2014partial}.
\end{rem}
}

To analyze the asymptotic behaviors of $V_n(k)$, we make the following moment assumption. 
\begin{ass}\label{ass_moment}
 The marginal distribution  $\nu(\cdot) \in M_4(\Omega)$.
\end{ass}

The following theorem derives the limiting distribution of \eqref{Vn_U} by exploiting the properties of U-statistics (e.g. \cite{lee1990u}). We emphasize here that even under $H _0$,  $V_n(k)$ is not a conventional U-statistic with kernels applied to  i.i.d. samples. For example, $Z_{t}^{(k)}=(X_t,X_{t-k})$ and $Z_{t+k}^{(k)}=(X_{t+k},X_t)$ are not independent of each other. Therefore, more delicate technical treatments are required. 
\begin{thm}\label{thm_fix}
Under $H_0$, and suppose Assumption \ref{ass_moment} holds. Then,  fix $K\geq 1$, as $n\to\infty$,
$$\left\{(n-k)V_n(k)\right\}_{k=1}^K\to_d \left\{\xi_k\right\}_{k=1}^K,$$
where $\xi_k=_d\sum_{\ell=1}^{\infty}\lambda_{\ell}\{[G_{\ell}^{(k)}]^2-1\}$. Here  $\{G_{\ell}^{(k)}\}_{\ell=1,2,\cdots,\infty; k=1,\cdots,K}$ is a sequence of i.i.d. standard Gaussian random variables,
and $\{\lambda_{\ell}\}_{\ell=1}^{\infty}$  and $\{e_{\ell}(\cdot)\}_{\ell=1}^{\infty}$    are   sequences of nonzero eigenvalues and orthonormal eigenfunctions corresponding to 
\begin{equation}\label{eq_eigen}
  \lambda_{\ell} e_{\ell}(z) = \mathbb{E}[d_{\nu}(x,X)d_{\nu}(y,Y)]e_{\ell}(z),
\end{equation}
where   $z=(x,y)\in\Omega^2$ and $Z=(X,Y)\in\Omega^2$ with $X$ and $Y$ being independent copies of $X_t$.
\end{thm}

\Cref{thm_fix} improves  \cite{zhou2012measuring} under $H_0$  not only by extending the result in Euclidean space to more general metric spaces but also by providing joint convergence of sample  ADCVs. In addition, \Cref{thm_fix}  is  crucial  for proving  the asymptotics of spectrum based test below. In practice,    note that $\{\xi_k\}_{k=1}^K$ is a sequence of centered mixture of  i.i.d. $\chi^2(1)$ random variables, which are non-pivotal. Below we adopt a wild  bootstrap  method to approximate the limiting null distributions, and details are deferred to \Cref{Sec:boot}.

\subsection{Generalized spectral test}
For our testing purpose,  it is natural to combine  ADCVs at all lags. We propose the following  generalized spectral density,
\begin{equation}\label{gspec}
    f(\zeta)= (2\pi)^{-1} \sum_{k=-\infty}^{\infty} V(k) e^{-ik\zeta},\quad \zeta\in [-\pi,\pi],
\end{equation}
and  generalized spectral distribution function,  $F(\zeta)=\int_{0}^{\zeta} f(z)\mathrm{d}z,~\zeta\in[0,\pi].$
Clearly,  $f(\cdot)$ (or $F(\cdot)$) is motivated by the spectral density (or distribution) function in Euclidean space and Hilbert space,  {where  the classical   spectral density function is defined by replacing $V(k)$ in \eqref{gspec} by $\gamma(k)$,
and $\gamma(k)=\mathrm{Cov}(X_t,X_{t-k})$ is the auto-covariance  at lag $k$. 

Our generalized spectral density (or function) serves as a ``generalization" of the classical spectral density (or function) in two significant aspects.  First, we aim to test serial independence, which is a stronger notion than serial uncorrelatedness.  In the conventional spectral density function, $\gamma(k)$ is the auto-covariance measuring linear correlation between $X_t$ and its lag-$k$ observation $X_{t-k}$, which implies that spectrum based tests can only capture linear serial dependence in the second-order structure.  By contrast, if we consider $V(k)$, the auto-distance covariance at lag $k$, we can additionally measure the lag $k$ nonlinear dependence.  Second, the definition of auto-covariance requires certain algebraic operation such as subtraction and inner product in the conventional Euclidean or Hilbert space, whereas the  auto-distance covariance does not, making it applicable to more general metric spaces.  } 



Under $H_0$, it holds that $F(\zeta)=V(0)\Psi_0(\zeta),$
where $\Psi_k(\zeta)=\sin(k\zeta)/(k\pi)\mathbf{1}(k\neq 0)+\zeta/(2\pi)\mathbf{1}(k=0)$, so that  we can construct a test by comparing the empirical estimator of  $F(\zeta)-V(0)\Psi_0(\zeta)$ and zero.
In particular, we define the process 
$$
S_n(\zeta) = \sum_{k=1}^{n-4}(n-k)V_n(k)\Psi_k(\zeta),\quad \zeta \in[0,\pi],
$$
and consider the following  Cram\'er von–Mises
(CvM) type statistic 
$$
 \textsc{CvM}_n = \int_{0}^{\pi} S_n^2(\zeta)\mathrm{d}\zeta.
$$ 
Similar  CvM type statistics have been adopted in Euclidean space and Hilbert space for testing second-order white noise or martingale differences, see, e.g. \cite{deo2000spectral}, \cite{ escanciano2006generalized},  \cite{shao2011bootstrap}, \cite{zhang2016white} among many others. 

\begin{rem}
One could also consider   KS (Kolmogorov–Smirnov) type statistics of the form 
\begin{equation}\label{KS}
\textsc{KS}_n=\sup_{\zeta\in[0,\pi]}|S_n(\zeta)|.
\end{equation}
However, the ``$\sup$" functional is no longer a continuous map in $L_2$ Hilbert space, and therefore the process convergence result below and the continuous mapping theorem do not lead to the asymptotic null distribution of 
KS test statistic, and a rigorous theoretical investigation is beyond the scope of this paper.  Nevertheless, our numerical studies in Section \ref{Sec:simu} suggest that $\textsc{KS}_n$ also delivers satisfactory performance.
\end{rem}

\begin{rem}
We briefly discuss the difference of our tests with the distance covariance based tests by \cite{fokianos2017consistent,fokianos2018testing} in  Euclidean space.
First, their test is designed for Euclidean valued time series, while ours has much more broad applications in other metric spaces. Second, their test is founded on the basis of  \cite{hong1999hypothesis} by noticing the connection between the characteristic function based definition of  distance covariance (c.f. \Cref{dfn_dcovE}) and the dependence metric defined in \cite{hong1999hypothesis} in terms of the weighting function. However their approach is no longer valid in metric space.  Third, their test statistic is based on  a kernel weighted sum of empirical ADCVs at a set of lags. Therefore, they require an additional tuning parameter to truncate the kernel, which is typically known to affect the convergence rate and hence power.  On the contrary, our test is tuning-free, and is more powerful even in the Euclidean setting, as demonstrated in Appendix \ref{Sec:Simul-Uni} in the supplement.  
\end{rem}

    

Our next theorem obtains the weak convergence result for $S_n(\zeta)$ {in the sense of $L_2$ metric.}  
\begin{thm}\label{thm_spec}
Under $H_0$, suppose  Assumption \ref{ass_moment} holds,  then
$$
\{S_n(\zeta)\}_{\zeta\in[0,\pi]}\Rightarrow \{S(\zeta)\}_{\zeta\in[0,\pi]}, \quad  \text{in  } L_2[0,\pi],
$$
where $\{S(\zeta)\}_{\zeta\in[0,\pi]}=_d \left\{\sum_{k=1}^{\infty}\xi_k\Psi_k(\zeta)\right\}_{\zeta\in[0,\pi]}$ with $\{\xi_k\}$ defined in \Cref{thm_fix}.
\end{thm}
Based on \Cref{thm_spec} and {note that the integral functional is continuous},  the following corollary is a natural consequence of continuous mapping theorem.
\begin{cor}
Under $H_0$, suppose  Assumption \ref{ass_moment} holds, then 
$$  \textsc{CvM}_n\to_d \int_{0}^{\pi} S^2(\zeta)\mathrm{d}\zeta.$$
\end{cor}

Therefore, given significance level $\alpha\in(0,1)$, we reject $H_0$ if  $\textsc{CvM}_n\geq \textsc{CvM}(1-\alpha),$ where  $\textsc{CvM}(1-\alpha)$ denotes the $1-\alpha$ quantile of $\int_{0}^{\pi} S^2(\zeta)\mathrm{d}\zeta$.  For practical implementation, we need to invoke the  bootstrap method in Section \ref{Sec:boot} to approximate the above critical value.

To study the behaviors of the tests under $H_a$, we impose the following $\beta$-mixing (absolutely regularity) conditions on $\{X_t\}_{t\in\mathbb{Z}}$. Basically, it requires weak temporal dependence of the process and similar assumptions are adopted in \cite{fokianos2017consistent,fokianos2018testing} and \cite{davis2018applications} for distance covariance based inference in  time series. Note that weak temporal dependence is also required in \cite{zhou2012measuring} using physical dependence measures \citep{wuwb2005}.
\begin{ass}\label{ass_mixing}
 $\{X_t\}_{t=1}^n$ is a stationary $\beta$-mixing process, with marginal distribution $\nu(\cdot) \in M_{2+2\delta}(\Omega)$ for some $\delta>0$. Furthermore, the mixing coefficient satisfies that $\beta(n)=O(n^{-r})$ for some $r>0$.
\end{ass}

\Cref{thm_consis}  below derives the asymptotics of the test statistics under $H_a$.
\begin{thm}\label{thm_consis}
    Suppose Assumption \ref{ass_mixing} hold, we have  $V_n(k)\to_{a.s.} V(k)$ for $k=1,2,\cdots$. Under $H_a$,  as $n\to\infty$,
    $ n^{-1}\textsc{CvM}_n\to_p \sum_{k=1}^{\infty}V^2(k)\int_{0}^{\pi} \Psi_k^2(\zeta)\mathrm{d}\zeta.$
\end{thm}
Clearly, under $H_a,$ we have $V(k)>0$ for some $k\neq 0$, which  implies that our test has asymptotically power 1 under $H_a$.


\section{Bootstrap}\label{Sec:boot}
In this section, we use wild bootstrap  method to approximate the limiting null distribution of $\textsc{CvM}_n$, and show its asymptotic validity.   

\begin{alg}\label{alg1}
~~
Let $\mathbf{w}=(w_1,\cdots,w_n)$ be a sequence of independent random variables following the Rademacher  distribution: 
$
\mathbb{P}(w_i=1)=\mathbb{P}(w_i=-1)=\frac{1}{2}.
$
\begin{enumerate}

\item  For each $k=1,\cdots,n-4$, define the doubly weighted bootstrapped version of $V_n$  by 
\begin{equation}\label{Vboot}
    V_n^{*}(k)^{(b)}= \frac{1}{(n-k)(n-k-3)}\sum_{k+1\leq i\neq j\leq n}w_i^{(b)}(k)\tilde{a}_{ij}\tilde{b}_{ij}w_j^{(b)}(k),
\end{equation}
where  $\tilde{a}_{ij}$ and $\tilde{b}_{ij}$ are $\mathcal{U}$ centering version of $a_{ij}=d(X_i,X_j)$ and $b_{ij}=d(X_{i-k},X_{j-k})$ respectively, and $\mathbf{w}^{(b)}(k)$ is an independent random realization of $\mathbf{w}$ with $w_i^{(b)}(k)$ being the $i$th element of $\mathbf{w}^{(b)}(k)$.

    \item Obtain the bootstrapped process 
    $$
    S_n^*(\zeta)=\sum_{k=1}^{n-4}(n-k)V_n^{*}(k)^{(b)}\Psi_k(\zeta),\quad \zeta\in[0,\pi],
    $$
    and compute 
$
 \textsc{CvM}_n^{*(b)} = \int_{0}^{\pi} [S_n^{*(b)}]^2(\zeta)\mathrm{d}\zeta.
$
    \item Repeat steps 1-2 $B$ times and collect  $\{ \textsc{CvM}_n^{*(b)}\}_{b=1}^B$.
    \item Denote the $\alpha$th upper percentile of $\{ \textsc{CvM}_n^{*(b)}\}_{b=1}^B$ by  $\textsc{CvM}_n^*(1-\alpha)$, and reject $H_0$ if  $ \textsc{CvM}_n> \textsc{CvM}_n^*(1-\alpha)$.
\end{enumerate}
\end{alg}

For each fixed $k$, the bootstrapped auto-distance covariance estimator $V_n^{*}(k)^{(b)}$ can be viewed as a realization of the independent wild bootstrap for degenerate U-statistics, see \cite{dehling1994random} and \cite{lee2020testing}. They are further combined to obtain the bootstrap realization of $\textsc{CvM}_n^{*(b)}$. 
Note here we use independent copies of wild bootstrap weights $\mathbf{w}(k)$ for different $k$'s. This is implied by \Cref{thm_fix}  that despite common components in constructing $\{V_n(k)\}_{k=1}^K$, their asymptotic  distributions $\{\xi_k\}_{k=1}^K$  are independent of each other under $H_0$.  

\begin{rem}
    {For our test targeting at the global null  of serial independence}, independent wild bootstrap is sufficient.   If one is interested in making statistical inference for $V_n(k)$ alone, potential serial dependence at other lag orders should be accounted for. {A possible extension along this direction would be to pinpoint the minimum lag at which dependencies exist by sequentially testing the nullity of $V_n(k)$.  To achieve this, the above wild bootstrap method should be adjusted accordingly. For example, \cite{leucht2013dependent} used    the dependent wild bootstrap in \cite{shao2010dependent} by constructing temporarily dependent weights $\{w_i\}_{i=1}^n$, as an extension of wild bootstrap in \cite{dehling1994random} for degenerate U-statistics.   See also  \cite{zhou2012measuring} for a subsampling based procedure. } 
\end{rem}

\begin{rem}
   {  Another approach to approximate the limiting null distribution is by using permutation. Although both bootstrap and permutation methods can be based on pre-calculated pairwise distance metrics $d(X_i, X_j)$ for $1 \leq i < j \leq n$, permutation is less computationally efficient. This is because in the $\mathcal{U}$-centering step, permutation method needs to re-calculate $\tilde{a}_{ij}$ (or $\tilde{b}_{ij}$) in \eqref{atilde}, which  involves the summation of $d(X_{\tau(i)},X_{\tau(j)})$ with $\tau(i)$ being the permutation order of $i$th original observation. In contrast, bootstrap method does not require such re-calculation, because $\tilde{a}_{ij}$ (or $\tilde{b}_{ij}$) is kept constant, as shown in \eqref{Vboot}, across all bootstrap statistics.}
\end{rem}

We then proceed to analyze the asymptotic behavior of the bootstrapped test statistics. We first introduce  the notion of \textit{convergence in distribution in probability}, see Definition 2 in \cite{li2003consistent}.   
\begin{defn}[convergence in distribution in probability]
    For a sequence of bootstrapped statistics $T_n^*$ which depends on the random sample $\{X_t\}_{t=1}^n$, we say that $T_n^*|(X_1,X_2,\cdots)$ converges to  $T|(X_1,X_2,\cdots)$ in distribution in probability  if for any subsequence $T_n'$,  there exists
a further subsequence $T_n''$ such that $T_n''|(X_1,X_2,\cdots)$ converges to $T|(X_1,X_2,\cdots)$ in distribution for almost every
sequence $(X_1,X_2,\cdots)$. The following notation is used,
$$T_n^* \to_{d^*} T,\quad \mbox{in probability.}$$
\end{defn}

\begin{thm}\label{thm_boot_fix}
Under $H_0$, suppose Assumptions \ref{ass_moment} holds. For $\mathbf{w}$ such that $\mathbb{E}w_i=0$, $\mathbb{E}w_i^2=1$ and $\mathbb{E}w_i^4<\infty$. Then, for any fixed $K\geq 1$, as $n\to\infty$,
$$\{(n-k)V_n^{*}(k)\}_{k=1}^K\to_{d*} \left\{\xi_k\right\}_{k=1}^K \quad \mbox{in probability,}$$
 where $\xi_k$ is defined in Theorem \ref{thm_fix}.
\end{thm}

\begin{thm}\label{thm_boot_test}
Under the same conditions of \Cref{thm_boot_fix}, 
$$
 \textsc{CvM}_n^*\to_{d^*} \int_{0}^{\pi}S^2(\zeta)\mathrm{d}\zeta, \quad\mbox{in probability,}
$$
where $\{S(\zeta)\}_{\zeta\in[0,\pi]}=_d \left\{\sum_{k=1}^{\infty}\xi_k\Psi_k(\zeta)\right\}_{\zeta\in[0,\pi]}$, with $\{\xi_k\}$ defined in \Cref{thm_fix}.  
{Furthermore, we have  $
\mathbb{P}( \textsc{CvM}_n> \textsc{CvM}_n^*(1-\alpha))\to \alpha.
$}
\end{thm}

Therefore, the wild bootstrap in \Cref{alg1} provides a consistent approximation of the limiting  null distribution of our test statistics. When $H_0$ is violated, we show that the bootstrapped tests have consistent powers.  

\begin{thm}\label{thm_boot_consis}
Suppose  \Cref{ass_mixing}  with $\delta=1$, then under $H_a$, for $\alpha\in(0,1)$, we have   $$
\mathbb{P}( \textsc{CvM}_n> \textsc{CvM}_n^*(1-\alpha))\to 1.
$$
\end{thm}

\section{Simulation studies}\label{Sec:simu}

In this section, we look into the finite sample performance of the proposed method on several different types of data objects. 
{Specifically, 
the performance on the functional time series in $L_2[0,1]$ is investigated in Section \ref{Sec:Simul-Func} whereas that on a sequence of covariance matrices can be found in Section \ref{Sec:Simul-Cov}. Finally, Section \ref{Sec:Simul-Dist} reports the simulation results on a sequence of univariate distributions. Additional results of time series in the Euclidean space are provided in \Cref{Sec:AddSimul} of the supplement.}

In addition to the wild bootstrap method (denoted by \texttt{Boot}) introduced in the previous section, we also apply the standard permutation test (denoted by \texttt{Permt}) to simulate the critical value for the proposed \textsc{CvM} test. Furthermore, we outline the results of \textsc{KS} in  \eqref{KS} using the same bootstrapped or permutated samples of $S_n^*(\cdot)$.
For each data structure, both bootstrap and permutation conduct $R=300$ replicates. {The empirical rejection rates in \Cref{Sec:Simul-Func} and \Cref{Sec:Simul-Dist} are averaged over $M=2500$ Monte Carlo replicates, whereas those in \Cref{Sec:Simul-Cov} are based on $M=1000$ replicates due to the computational expenses.}

\subsection{Functional time series}\label{Sec:Simul-Func}

We consider the scenario where $\Omega=\{f:[0,1]\mapsto\br, \int_{0}^{1}f^2(\tau)d\tau<\infty\}$ with $L_2$ metric, i.e. $d(f,g) = \bigp{\int_{0}^{1} (f(t)-g(t))^2 dt}^{1/2}$ for any $f,g\in\Omega$. Let $\cb:=\{B(\tau):\tau\in[0,1]\}$ denote the standard Brownian motion, and we mimic the setups in \cite{zhang2016white} to generate $\{Y_t(\tau): \tau\in[0,1]\}_{t=1}^{n}$ {from the following data generating processes (DGPs):}
\begin{enumerate}[label=(\roman*)]
    \item BM: $Y_t(\tau) = \varepsilon_t(\tau)$ for $\tau\in[0,1]$, where $\{\varepsilon_t(\tau)\}_{t=1}^{n} \stsim{i.i.d.} \cb$;
    \item BB: $Y_t(\tau) = \varepsilon_t(\tau) - \tau \varepsilon_t(1)$ for $\tau\in[0,1]$, where $\{\varepsilon_t(\tau)\}_{t=1}^{n} \stsim{i.i.d.} \cb$; 
    {
    \item FARCH: $Y_t(\tau_1) = \epsilon_t(\tau_1) \sqrt{\tau_1 + \int_{0}^{1} \rho \exp\lrp{(\tau_1^2+\tau_2^2)/2} Y_{t-1}^2(\tau_2) d\tau_2}$ for $\tau\in[0,1]$, where $\{\varepsilon_t(\tau)\}_{t=1}^{n} \stsim{i.i.d.} \cb$. We consider $\rho \in \{0.3418, 0.6153\}$, which corresponds to the cases when the Hilbert–Schmidt norm of the Gaussian kernel equals 0.5 and 0.9, respectively.
    \item FNMA: $Y_t(\tau) = \varepsilon_t(\tau) \varepsilon_{t-1}(\tau)$ for $\tau\in[0,1]$, where $\{\varepsilon_t(\tau)\}_{t=1}^{n}$ is either BM or BB.}
    \item FAR: $Y_t(\tau_1) = \int_{0}^{1} \ck(\tau_1,\tau_2) Y_{t-1}(\tau_2) d\tau_2 + \varepsilon_t(\tau_1)$, where $\{\varepsilon_t(\tau)\}_{t=1}^{n}$ is generated by either BM or BB, and we consider the Gaussian kernel $\ck_G(\tau_1,\tau_2) = c_G \exp\lrp{(\tau_1^2+\tau_2^2)/2}$ and the Wiener kernel $\ck_W(\tau_1,\tau_2) = c_W \min\{\tau_1,\tau_2\}$. For each kernel, we choose the coefficient $c$ so that the corresponding Hilbert–Schmidt norm equal to 0.3, i.e. $c_G=0.2051$ for $\ck_G$ and $c_W=0.7346$ for $\ck_W$.
\end{enumerate}

Among the existing tests that target at the serial dependence within a sequence of functional data, we consider those proposed in \cite{gabrys2007portmanteau}, \cite{horvath2013test} and \cite{zhang2016white} for comparison. In particular,
\begin{enumerate}[label=(\roman*)]
    \item \cite{gabrys2007portmanteau} (denoted by \texttt{GK}) proposed the portmanteau test based on the functional principal component analysis (fPCA). We consider the lag truncation number $h\in\{1,3,5\}$ and select the smallest $P$ such that the cumulative variation explained by the first $P$ principal components is above 90\%.
    
    \item \cite{horvath2013test} (denoted by \texttt{HHHR}) presented an independence test based on the empirical correlation functions. We aggregate the recommendations of \cite{horvath2013test} and \cite{zhang2016white} to consider the lag truncation number $h\in\{5,10,30,50\}$. Each integral involved in \texttt{HHHR} is approximated by the Riemann sum with $p=100$ points. 
    
    \item \cite{zhang2016white} proposed a Cram\'er–von Mises type test (denoted by \texttt{Zhang}) based on the functional periodogram. The critical value of \texttt{Zhang} is approximated using block bootstraps with $R=300$ replicates, and we consider the block size $b\in\{1,5,10\}$ as well as the optimal size $b^{\ast}$ chosen by the minimal volatility method. All the integrals are approximated by the Riemann sum with $p=50$ points to ease the computational burden.
\end{enumerate}

We fix $n=200$ and generate the data on a grid of $T=1000$ equally spaced points in $[0,1]$ for each functional observation. We use Fourier and B-splines (order 4) with $B=20$ basis functions to obtain the functional data. 

\begin{table}[!h]
    \caption{Empirical rejection rate of functional time series when $n=200$}
    \label{Tab:Func-R1}
    \centering
    \scalebox{0.7}{
    \begin{tabular}{c|c||cc|cc||ccc||cccc||cccc}
    \hline\hline
        \multicolumn{2}{c||}{\multirow{2}{*}{DGP}} & \multicolumn{4}{c||}{Proposed} & \multicolumn{3}{c||}{\texttt{GK}} & \multicolumn{4}{c||}{\texttt{HHHR}} & \multicolumn{4}{c}{\texttt{Zhang}} \\ \cline{3-17}
        \multicolumn{2}{c||}{} & CvM-P & CvM-B & KS-P & KS-B & $h=1$ & $h=3$ & $h=5$ & $h=5$ & $h=10$ & $h=30$ & $h=50$ & $b=1$ & $b=5$ & $b=10$ & $b^{\ast}$ \\ \hline
        \multirow{2}{*}{BM} & \texttt{F} & 0.059 & 0.057 & 0.059 & 0.060 & 0.048 & 0.052 & 0.040 & 0.068 & 0.064 & 0.064 & 0.071 & 0.053 & 0.057 & 0.062 & 0.062 \\
        & \texttt{B} & 0.055 & 0.058 & 0.057 & 0.060 & 0.047 & 0.050 & 0.041 & 0.068 & 0.064 & 0.063 & 0.071 & 0.055 & 0.054 & 0.060 & 0.062 \\ \hline
        \multirow{2}{*}{BB} & \texttt{F} & 0.052 & 0.051 & 0.050 & 0.050 & 0.040 & 0.044 & 0.034 & 0.057 & 0.054 & 0.068 & 0.072 & 0.043 & 0.038 & 0.036 & 0.043 \\
        & \texttt{B} & 0.051 & 0.054 & 0.047 & 0.051 & 0.039 & 0.042 & 0.032 & 0.057 & 0.054 & 0.068 & 0.072 & 0.046 & 0.038 & 0.037 & 0.044 \\ \hline
        \multirow{2}{*}{FARCH(0.5)} & \texttt{F} & 0.113 & 0.085 & 0.119 & 0.081 & 0.127 & 0.111 & 0.087 & 0.106 & 0.099 & 0.086 & 0.086 & 0.052 & 0.046 & 0.051 & 0.056 \\
        & \texttt{B} & 0.114 & 0.090 & 0.114 & 0.081 & 0.127 & 0.111 & 0.087 & 0.106 & 0.100 & 0.086 & 0.087 & 0.047 & 0.043 & 0.050 & 0.053 \\ \hline
        \multirow{2}{*}{FARCH(0.9)} & \texttt{F} & 0.296 & 0.194 & 0.282 & 0.173 & 0.288 & 0.279 & 0.240 & 0.213 & 0.164 & 0.107 & 0.090 & 0.045 & 0.037 & 0.041 & 0.045 \\
        & \texttt{B} & 0.294 & 0.184 & 0.275 & 0.174 & 0.287 & 0.278 & 0.241 & 0.213 & 0.163 & 0.108 & 0.090 & 0.043 & 0.037 & 0.038 & 0.046 \\ \hline
        \multirow{2}{*}{FNMA-BM} & \texttt{F} & 0.839 & 0.626 & 0.810 & 0.585 & 0.373 & 0.222 & 0.172 & 0.180 & 0.149 & 0.102 & 0.084 & 0.039 & 0.039 & 0.045 & 0.046 \\
        & \texttt{B} & 0.841 & 0.624 & 0.804 & 0.591 & 0.372 & 0.225 & 0.171 & 0.178 & 0.148 & 0.101 & 0.084 & 0.042 & 0.040 & 0.045 & 0.045 \\ \hline
        \multirow{2}{*}{FNMA-BB} & \texttt{F} & 0.774 & 0.518 & 0.730 & 0.496 & 0.512 & 0.269 & 0.188 & 0.211 & 0.164 & 0.113 & 0.110 & 0.044 & 0.028 & 0.026 & 0.035 \\
        & \texttt{B} & 0.774 & 0.518 & 0.730 & 0.494 & 0.492 & 0.262 & 0.185 & 0.208 & 0.164 & 0.110 & 0.108 & 0.042 & 0.027 & 0.028 & 0.033 \\ \hline
        \multirow{2}{*}{FAR-G-BM} & \texttt{F} & 0.984 & 0.984 & 0.981 & 0.981 & 0.964 & 0.832 & 0.726 & 0.921 & 0.822 & 0.662 & 0.559 & 0.973 & 0.970 & 0.958 & 0.973 \\
        & \texttt{B} & 0.985 & 0.984 & 0.981 & 0.980 & 0.961 & 0.828 & 0.728 & 0.921 & 0.818 & 0.656 & 0.557 & 0.972 & 0.968 & 0.957 & 0.972 \\ \hline
        \multirow{2}{*}{FAR-G-BB} & \texttt{F} & 0.984 & 0.985 & 0.978 & 0.981 & 0.972 & 0.827 & 0.692 & 0.866 & 0.741 & 0.540 & 0.453 & 0.957 & 0.942 & 0.924 & 0.944 \\
        & \texttt{B} & 0.985 & 0.983 & 0.979 & 0.980 & 0.966 & 0.818 & 0.682 & 0.866 & 0.741 & 0.541 & 0.454 & 0.960 & 0.938 & 0.932 & 0.943 \\ \hline
        \multirow{2}{*}{FAR-W-BM} & \texttt{F} & 0.973 & 0.972 & 0.971 & 0.970 & 0.919 & 0.786 & 0.678 & 0.918 & 0.864 & 0.718 & 0.636 & 0.984 & 0.979 & 0.968 & 0.983 \\
        & \texttt{B} & 0.974 & 0.973 & 0.966 & 0.971 & 0.920 & 0.788 & 0.680 & 0.920 & 0.865 & 0.719 & 0.635 & 0.984 & 0.974 & 0.975 & 0.980 \\ \hline
        \multirow{2}{*}{FAR-W-BB} & \texttt{F} & 0.968 & 0.969 & 0.966 & 0.964 & 0.933 & 0.706 & 0.550 & 0.851 & 0.715 & 0.520 & 0.441 & 0.950 & 0.914 & 0.905 & 0.931 \\
        & \texttt{B} & 0.974 & 0.971 & 0.967 & 0.971 & 0.936 & 0.712 & 0.562 & 0.857 & 0.719 & 0.521 & 0.444 & 0.953 & 0.931 & 0.913 & 0.939 \\ 
    \hline\hline    
    \end{tabular}}
\end{table}

Table \ref{Tab:Func-R1} includes the empirical rejection rate under each DGP, where {P and B are shorthands for \texttt{Permt} and \texttt{Boot} respectively}, and \texttt{F} and \texttt{B} denote the Fourier basis and B-splines respectively. 
Overall, the type of basis functions has limited influence on the finite sample performance for all the methods in Table \ref{Tab:Func-R1}. Under the null models BM and BB, the empirical size of \textsc{CvM} and \textsc{KS} are both close to 0.05, and the wild bootstrap has similar size accuracy with the permutation test. In contrast, the size accuracy of \texttt{GK} and \texttt{HHHR} depends on the lag truncation parameter $h$, and moderate size distortion can be observed with a larger number of lags. {We also note that \texttt{Zhang} achieves decent size accuracy with a smaller block size $b$, but has some size distortion with a large $b$.} 

{As for the power behaviors, our tests, especially with \texttt{Permt},   demonstrate dominating performance in FNMA and FAR models. While they are slightly inferior in FARCH(0.5), the performances get better when dependence level increases in FARCH(0.9).  
In comparison,  both \texttt{GK} and \texttt{HHHR} seem to be sensitive to truncation parameter selections,  leading to a considerable loss in power,  especially with a larger lag truncation parameter. 
Note that both FARCH and FNMA models have serial dependence but no serial correlation, thus corresponding to null models in \cite{zhang2016white}. Consequently, the power of \texttt{Zhang} against these models is close to the nominal level. We also note that   \texttt{GK} and \texttt{Zhang} can have comparable power with our tests but only when a smaller bandwidth (or block size) parameter is used. Overall, our tests are  among the best across all settings.}

\subsection{Covariance matrix}\label{Sec:Simul-Cov}

 Next, we consider the scenario where the observations are a sequence of covariance matrices $\{Y_t\}_{t=1}^{n} \in \br^{p\times p}$, which is relatively new in literature.


{We generate the time series from a conditional autoregressive Wishart model \citep{golosnoy2012conditional} with dimension $p \in \{2,5,8\}$. In particular, let $\cw_p(d,V)$ denote the Wishart distribution with degrees of freedom $d$ and the scale matrix $V$, and we generate $\{\varepsilon_t\}_{t=1}^{n} \stsim{i.i.d.} \cw_p(10,I_p) \in \br^{p\times p}$. Furthermore, we generate $Y_t = \frac{1}{10} \chol(\Sigma_t) \varepsilon_t \chol(\Sigma_t)$, where $\chol(\Sigma_t)$ denotes the lower-triangular Cholesky factor of $\Sigma_t$, and $\Sigma_t = CC^{\top} + \rho A Y_{t-1} A^{\top} + \rho B \Sigma_{t-1} B^{\top}$ with 
$A = 0.7 I_p,$
$B = 0.5 I_p$,
and $C = I_p$. 
Here we use the parameter $\rho\in\{0, 0.5, 75\}$ to control the temporal dependence and it is trivial that $\rho=0$ corresponds to the null model.  We use $\text{CAW}_p(\rho)$ to denote the model with dimension $p$ and parameter $\rho$.}


Note that in this case, we have $\Omega = \{A: \mbox{symmetric, positive definite matrix} \in \br^{p\times p}\}$. To evaluate the distance between two covariance matrices $A,B\in\Omega$, we consider the following metrics: (i) Euclidean metric: $d(A,B) = \|A-B\|_F$; (ii) log-Euclidean metric: $d(A,B) = \|\logm(A)-\logm(B)\|_F$, where $\logm(\cdot)$ denotes the logarithm of a matrix; (iii) Cholesky metric: $d(A,B) = \|\chol(A)-\chol(B)\|_F$, where $\chol(A)$ denotes the Cholesky decomposition of $A$; (iv) Riemann metric: $d(A,B) = \|\logm(A^{-1/2} B A^{-1/2})\|_F$, where  $A^{1/2}$ denotes the matrix square root of $A$. 


In view of the fact that there is little literature regarding this testing problem, we use the half vectorization to convert a $p\times p$ covariance matrix into a $\frac{1}{2}p(p+1)$-dimensional vector and apply the multivariate testing methods to the resulting time series. In addition to our proposed test with the Euclidean metric (denoted by \texttt{Vech-Euc}), {we also adopt the test statistic \texttt{mADCV} proposed in \cite{fokianos2018testing}, which can be implemented via the \texttt{R} package \texttt{dCovTS}.} In particular, \texttt{mADCV} is based on the auto-distance covariance matrix {and requires a kernel function $k(\cdot)$ and a bandwidth parameter $p_n$. We follow the setup of \cite{fokianos2018testing} to consider the truncated kernel (TC), the Daniell kernel (DAN), the Parzen kernel (PAR), and the Bartlett kernel (BAR).} Here, we set $p_n = \lrfl{3n^{\lambda}}$ with $\lambda\in\{1/10,2/10,3/10\}$. 

{According to Table \ref{Tab:Cov-R3}, our test achieves accurate size, regardless of the dimension $p$ and the choice of metric. Under the alternative, our proposed test with the Euclidean metric demonstrates the highest power when applied to either matrix-valued time series or half-vectorized multivariate time series. Interestingly, the  power increases as the dimension $p$ grows, indicating its advantage in the moderate-dimensional scenario. On the other hand, the log-Euclidean metric and the Riemann metric exhibit high power when $p$ is small, but they do show a slight power loss against moderate values of $p$. This discrepancy in performance could be attributed to the fact that different metrics excel at capturing distinct dependence structures.


In contrast, \texttt{mADCV} has a conservative size under the null and have noticeable power loss under the alternative, especially when $p$ is moderate and $\rho$ is small. When $p=5$ and $p=8$, the half-vectorized covariance matrix is converted to the multivariate time series in $\br^{15}$ and $\br^{36}$, respectively. Note that \texttt{mADCV} involves the estimation of the long-run covariance matrix, and the estimation tends to be less accurate as the dimension grows, which may explain why \texttt{mADCV} worsens against growing $p$.}


\begin{table}[!h]
    \caption{Empirical rejection rate of covariance matrix time series generated by $\text{CAW}_p(\rho)$ model}
    \label{Tab:Cov-R3}
    \centering
    \scalebox{0.85}{
    \begin{tabular}{c|c|c||ccc||ccc||ccc}
    \hline\hline
        \multicolumn{3}{c||}{\multirow{2}{*}{Method}} & \multicolumn{3}{c|}{$p=2$} & \multicolumn{3}{c|}{$p=5$} & \multicolumn{3}{c}{$p=8$} \\ \cline{4-12}
        \multicolumn{3}{c||}{} & $\rho=0$ & $\rho=0.5$ & $\rho=0.75$ & $\rho=0$ & $\rho=0.5$ & $\rho=0.75$ & $\rho=0$ & $\rho=0.5$ & $\rho=0.75$ \\ \hline
        \multirow{20}{*}{Proposed} & \multirow{4}{*}{Euc} & CvM-P & 0.043 & 0.419 & 0.997 & 0.067 & 0.570 & 1.000 & 0.042 & 0.745 & 1.000 \\
        & & CvM-B & 0.044 & 0.412 & 0.999 & 0.069 & 0.558 & 1.000 & 0.046 & 0.744 & 1.000 \\
        & & KS-P & 0.045 & 0.403 & 0.997 & 0.061 & 0.529 & 1.000 & 0.046 & 0.681 & 1.000 \\
        & & KS-B & 0.048 & 0.394 & 0.997 & 0.072 & 0.517 & 1.000 & 0.053 & 0.694 & 1.000 \\ \cline{2-12}
        & \multirow{4}{*}{log-Euc} & CvM-P & 0.051 & 0.347 & 0.997 & 0.059 & 0.309 & 1.000 & 0.061 & 0.256 & 1.000 \\
        & & CvM-B & 0.049 & 0.349 & 0.996 & 0.058 & 0.316 & 1.000 & 0.061 & 0.254 & 1.000 \\
        & & KS-P & 0.048 & 0.340 & 0.993 & 0.060 & 0.292 & 1.000 & 0.063 & 0.226 & 1.000 \\
        & & KS-B & 0.051 & 0.339 & 0.994 & 0.065 & 0.294 & 1.000 & 0.068 & 0.223 & 1.000 \\ \cline{2-12}
        & \multirow{4}{*}{Chol} & CvM-P & 0.052 & 0.387 & 0.995 & 0.045 & 0.340 & 1.000 & 0.059 & 0.348 & 1.000 \\
        & & CvM-B & 0.052 & 0.378 & 0.996 & 0.045 & 0.339 & 1.000 & 0.065 & 0.356 & 1.000 \\
        & & KS-P & 0.046 & 0.355 & 0.995 & 0.049 & 0.314 & 1.000 & 0.061 & 0.305 & 1.000 \\
        & & KS-B & 0.050 & 0.356 & 0.995 & 0.051 & 0.321 & 1.000 & 0.065 & 0.315 & 1.000 \\ \cline{2-12}
        & \multirow{4}{*}{Riemann} & CvM-P & 0.053 & 0.346 & 0.996 & 0.062 & 0.316 & 1.000 & 0.058 & 0.290 & 1.000 \\
        & & CvM-B & 0.053 & 0.345 & 0.995 & 0.061 & 0.314 & 1.000 & 0.062 & 0.287 & 1.000 \\
        & & KS-P & 0.051 & 0.342 & 0.995 & 0.068 & 0.293 & 1.000 & 0.058 & 0.262 & 1.000 \\
        & & KS-B & 0.053 & 0.340 & 0.994 & 0.068 & 0.298 & 1.000 & 0.070 & 0.259 & 1.000 \\ \cline{2-12}
        & \multirow{4}{*}{Vech-Euc} & CvM-P & 0.048 & 0.399 & 0.994 & 0.059 & 0.527 & 1.000 & 0.044 & 0.703 & 1.000 \\
        & & CvM-B & 0.049 & 0.389 & 0.994 & 0.062 & 0.517 & 1.000 & 0.042 & 0.697 & 1.000 \\
        & & KS-P & 0.052 & 0.376 & 0.992 & 0.058 & 0.500 & 1.000 & 0.046 & 0.635 & 1.000 \\
        & & KS-B & 0.049 & 0.374 & 0.993 & 0.061 & 0.491 & 1.000 & 0.047 & 0.639 & 1.000 \\ \hline
        \multirow{12}{*}{\texttt{mADCV}} & \multirow{3}{*}{TC} & $\lambda=1/10$ & 0.034 & 0.146 & 0.899 & 0 & 0 & 0.898 & 0 & 0 & 0.666 \\
        & & $\lambda=2/10$ & 0.036 & 0.131 & 0.871 & 0 & 0.001 & 0.931 & 0 & 0 & 0.895 \\
        & & $\lambda=3/10$ & 0.037 & 0.119 & 0.822 & 0 & 0.003 & 0.963 & 0 & 0 & 0.983 \\ \cline{2-12}
        & \multirow{3}{*}{DAN} & $\lambda=1/10$ & 0.029 & 0.253 & 0.984 & 0 & 0.002 & 0.980 & 0 & 0 & 0.843 \\
        & & $\lambda=2/10$ & 0.029 & 0.209 & 0.971 & 0 & 0.003 & 0.973 & 0 & 0 & 0.915 \\
        & & $\lambda=3/10$ & 0.034 & 0.179 & 0.939 & 0 & 0.003 & 0.979 & 0 & 0 & 0.974 \\ \cline{2-12}
        & \multirow{3}{*}{PAR} & $\lambda=1/10$ & 0.026 & 0.243 & 0.983 & 0 & 0.003 & 0.977 & 0 & 0 & 0.837 \\
        & & $\lambda=2/10$ & 0.026 & 0.190 & 0.964 & 0 & 0.002 & 0.969 & 0 & 0 & 0.885 \\
        & & $\lambda=3/10$ & 0.037 & 0.175 & 0.932 & 0 & 0.003 & 0.977 & 0 & 0 & 0.974 \\ \cline{2-12}
        & \multirow{3}{*}{BAR} & $\lambda=1/10$ & 0.026 & 0.257 & 0.985 & 0 & 0.005 & 0.979 & 0 & 0 & 0.744 \\
        & & $\lambda=2/10$ & 0.026 & 0.213 & 0.976 & 0 & 0.004 & 0.981 & 0 & 0 & 0.890 \\
        & & $\lambda=3/10$ & 0.030 & 0.200 & 0.958 & 0 & 0.003 & 0.981 & 0 & 0 & 0.953 \\ 
    \hline\hline
    \end{tabular}}    
\end{table}

\subsection{Univariate distribution}\label{Sec:Simul-Dist}

Lastly, we consider the case where $\Omega$ is the set of cumulative distribution function{s} of a random variable that takes value from $[0,1]$. In this case, we mimic the settings in \cite{zhu2021autoregressive} to generate a sequence of CDFs $\{Y_t(x):x\in[0,1]\}_{t=1}^{n}$, with $n=200$, from an autoregressive transport model of order $p$ denoted by ATM(p). 

Specifically, let $g(x)$ denote the natural cubic spline passing through points $(0,0)$, $(0.33,0.2)$, $(0.66,0.8)$, $(1,1)$, and we generate a random sample $\{\xi_t\}_{t=1}^{n}$ from a uniform distribution over $(-1,1)$. For $x\in\cs:=[0,1]$, we define $h(x) = \frac{1}{2} \lrp{(1-\xi_t)g(x) + (1+\xi_t)x}$ and $ \varepsilon_t(x) = \frac{1}{2} \lrp{(1+\xi_t)g(h^{-1}(x)) + (1 {-} \xi_t) h^{-1}(x)}$. Then we generate a sequence of quantile functions $\{T_t\}_{t=1}^{n}$ by (i) ATM(0): $T_t = \varepsilon_t$; (ii) ATM(1): $T_t = \beta_1 \odot T_{t-1} \oplus \varepsilon_t$; (iii) $T_t = \beta_4 \odot T_{t-4} \oplus \beta_3 \odot T_{t-3} \oplus \beta_2 \odot T_{t-2} \oplus \beta_1 \odot T_{t-1} \oplus \varepsilon_t$, where $(T_1 \oplus T_2)(x) = T_2 \circ T_1(x) = T_2(T_1(x))$, and for $\beta\in[-1,1]$,
\begin{equation*}
    {\beta} \odot T(x) 
  = \left\{
      \begin{array}{ll}
          x+\beta(T(x)-x), & 0<\beta\le1  \\
          x, & \beta=0 \\
          x+\beta(x-T^{-1}(x)), & -1\le\beta<0
      \end{array}
    \right.
\end{equation*}
 For ATM(1), we consider $\beta_1=0.5\rho$ with $\rho\in\{0.2,0.5,1\}$, whereas for ATM(4), we consider $\beta=\rho(0.2,-0.5,0.1,-0.3)^{\top}$ with $\rho\in\{0.5,0.8,1\}$. Based on the simulated quantile functions $\{T_t(x): x\in\cs\}_{t=1}^{n}$, we generate $\{Y_t(x):x\in[0,1]\}_{t=1}^{n}$ by $Y_t(x) = T_t^{-1}(x)$.

For the metric space $\Omega$, we consider several metrics to evaluate the distance between two distributions. In particular, we use $F,G$ to denote two CDFs, use $F^{-1},G^{-1}$ to denote the corresponding quantile functions, and use $f,g$ to denote the corresponding density functions, then the metrics are given by (i) 1-Wasserstein (\texttt{W1}): $d_{W1}(F,G) = \int_{0}^{1} \lrabs{F^{-1}(t)-G^{-1}(t)} \mathrm{d}t$; (ii) 2-Wasserstein (\texttt{W2}): $d_{W2}(F,G) = \bigp{\int_{0}^{1} \bigp{F^{-1}(t)-G^{-1}(t)}^2 \mathrm{d}t}^{1/2}$; (iii) Kolmogorov–Smirnov (\texttt{KS}): $d_{KS}(F,G) = \sup_{x\in\cs} \lrabs{F(x)-G(x)}$; (iv) Kullback–Leibler (\texttt{KL}): $d_{KL}(F,G) = \frac{1}{2} \bigp{\int_{x\in\cs} \log\frac{F(\mathrm{d}x)}{G(\mathrm{d}x)} F(\mathrm{d}x) + \int_{x\in\cs} \log\frac{G(\mathrm{d}x)}{F(\mathrm{d}x)} G(\mathrm{d}x)}$;
(v) Itakura–Saito (IS): $d_{IS}(F,G) = \frac{1}{2|\cs|}\int_{x\in\cs}\lrp{\frac{f(x)}{g(x)} - \log\frac{f(x)}{g(x)} - 1}\mathrm{d}x + \frac{1}{2|\cs|}\int_{x\in\cs}\lrp{\frac{g(x)}{f(x)} - \log\frac{g(x)}{f(x)} - 1}\mathrm{d}x$; (vi) Log-Spectral (LS): $d_{LS}(F,G) = \lrp{\frac{1}{|\cs|} \int_{x\in\cs} \lrp{10 \log_{10}\frac{f(x)}{g(x)}}^2 \mathrm{d}x}^{1/2}$.

In practice, we compute the Wasserstein distances directly using the simulated quantile functions $\{T_t(x)\}_{t=1}^{n}$, and compute all the other distances using \texttt{R} package \texttt{seewave} with the probability functions estimated from the simulated CDFs.
To the best of our knowledge, there is no existing literature for testing the serial dependence within a sequence of univariate distributions, hence we only present the empirical rejection rates of the proposed method in Table \ref{Tab:Dist-3}. 

\begin{table}[!h]
    \caption{Empirical rejection rate of univariate distributions when $n=200$}
    \label{Tab:Dist-3}
    \centering
    \scalebox{0.8}{
    \begin{tabular}{c|c|c|cc|cc|cc|cc|cc|cc}
    \hline\hline
    \multicolumn{3}{c|}{\multirow{2}{*}{}} & \multicolumn{2}{c|}{\texttt{W1}} & \multicolumn{2}{c|}{\texttt{W2}} & \multicolumn{2}{c|}{\texttt{KS}} & \multicolumn{2}{c|}{\texttt{KL}} & \multicolumn{2}{c|}{\texttt{IS}} & \multicolumn{2}{c}{\texttt{LS}} \\ \cline{4-15}
    \multicolumn{3}{c|}{} & Permt & Boot & Permt & Boot & Permt & Boot & Permt & Boot & Permt & Boot & Permt & Boot \\ \hline
    \multicolumn{2}{c|}{\multirow{2}{*}{ATM(0)}} & \textsc{CvM} & 0.049 & 0.052 & 0.049 & 0.050 & 0.050 & 0.050 & 0.050 & 0.050 & 0.047 & 0.045 & 0.050 & 0.049 \\
    \multicolumn{2}{c|}{} & \textsc{KS} & 0.049 & 0.052 & 0.050 & 0.051 & 0.050 & 0.057 & 0.049 & 0.050 & 0.049 & 0.050 & 0.053 & 0.049 \\ \hline
    \multirow{6}{*}{ATM(1)} & \multirow{2}{*}{$\rho=0.2$} & \textsc{CvM} & 0.186 & 0.183 & 0.178 & 0.180 & 0.178 & 0.173 & 0.239 & 0.234 & 0.272 & 0.262 & 0.203 & 0.201 \\
    & & \textsc{KS} & 0.172 & 0.176 & 0.175 & 0.171 & 0.164 & 0.164 & 0.224 & 0.228 & 0.264 & 0.253 & 0.188 & 0.186 \\ \cline{2-15}
    & \multirow{2}{*}{$\rho=0.5$} & \textsc{CvM} & 0.624 & 0.623 & 0.622 & 0.624 & 0.603 & 0.595 & 0.745 & 0.732 & 0.814 & 0.789 & 0.696 & 0.689 \\
    & & \textsc{KS} & 0.607 & 0.604 & 0.599 & 0.603 & 0.579 & 0.575 & 0.728 & 0.716 & 0.800 & 0.782 & 0.678 & 0.676 \\ \cline{2-15}
    & \multirow{2}{*}{$\rho=1.0$} & \textsc{CvM} & 1.000 & 1.000 & 1.000 & 1.000 & 1.000 & 1.000 & 1.000 & 1.000 & 0.970 & 0.978 & 1.000 & 1.000 \\
    & & \textsc{KS} & 1.000 & 1.000 & 1.000 & 1.000 & 1.000 & 1.000 & 1.000 & 1.000 & 0.970 & 0.977 & 1.000 & 1.000 \\ \hline
    \multirow{6}{*}{ATM(4)} & \multirow{2}{*}{$\rho=0.5$} & \textsc{CvM} & 0.176 & 0.180 & 0.178 & 0.176 & 0.173 & 0.171 & 0.240 & 0.245 & 0.279 & 0.265 & 0.186 & 0.186 \\
    & & \textsc{KS} & 0.184 & 0.188 & 0.188 & 0.186 & 0.175 & 0.180 & 0.253 & 0.253 & 0.289 & 0.280 & 0.202 & 0.198 \\ \cline{2-15}
    & \multirow{2}{*}{$\rho=0.8$} & \textsc{CvM} & 0.730 & 0.732 & 0.732 & 0.737 & 0.717 & 0.714 & 0.846 & 0.834 & 0.862 & 0.824 & 0.762 & 0.759 \\
    & & \textsc{KS} & 0.735 & 0.736 & 0.742 & 0.745 & 0.724 & 0.718 & 0.855 & 0.840 & 0.870 & 0.838 & 0.756 & 0.750 \\ \cline{2-15}
    & \multirow{2}{*}{$\rho=1.0$} & \textsc{CvM} & 0.996 & 0.996 & 0.997 & 0.997 & 0.995 & 0.995 & 0.998 & 0.998 & 0.996 & 0.988 & 0.998 & 0.997 \\
    & & \textsc{KS} & 0.994 & 0.996 & 0.996 & 0.996 & 0.995 & 0.996 & 0.998 & 0.999 & 0.997 & 0.989 & 0.998 & 0.997 \\
    \hline\hline
    \end{tabular}}
\end{table}

The overall pattern shown in Table \ref{Tab:Dist-3} generally matches those observed in \Cref{Sec:Simul-Func} and \Cref{Sec:Simul-Cov}. Specifically, the type of test statistic and the method used to approximate the critical value have little influence on both the size and the power. Among all the metrics examined in Table \ref{Tab:Dist-3}, we observe no significant disparities in terms of the size accuracy, whereas the Itakura-Saito metric and the Kullback-Leibler metric exhibit the dominating power in most cases.


\section{Applications}\label{Sec:app}
In this section, we demonstrate the versatility of the proposed tests in {two}  real applications.
\subsection{Cumulative intraday returns}
Modelling  financial returns using functional data approach has attracted much attention in the literature, see e.g. \cite{aue2017functional}, \cite{shang2017forecasting}, \cite{cerovecki2019functional}. In the first application, we apply the proposed tests to the cumulative intraday returns (CIDRs) of two stocks of  International Business Machines Corporation (IBM)  and McDonald’s Corporation (MCD) in  2007. Following \cite{gabrys2010tests}, we  define the CIDRs as 
$$
X_t(s)=100\{\log[P_t(s)/P_t(0)]\},\quad  s\in[0,1],  
$$
where $P_t(s)$ denote the stock price at rescaled trading period $s$, and $P_t(0)$ denote the market opening price.   The daily trading data is collected at 1-minute frequency with a total 390 data points per day. The data is then smoothed using 20 Fourier basis functions.  We consider the proposed CvM test and KS test, along with other competing methods, including \texttt{GK} at lag $h={1,3,5}$, \texttt{HHHR} with truncation lag $h={5,10,30,50}$, and \texttt{Zhang} with $b={1,5,10}$. The p-values of all tests are summarized in Table \ref{app:tab1}.
\begin{table}[H]
\caption{P-values for white noise testing in CIDRs of IBM and MCD in 2007. }
\centering
\label{app:tab1}
    \scalebox{0.85}{\begin{tabular}{ccccccccccccc}
\hline\hline
Test                    & CvM    & KS    & \texttt{GK}1 & \texttt{GK}3   & \texttt{GK}5   & \texttt{HHHR}5 & \texttt{HHHR}10  & \texttt{HHHR}30   & \texttt{HHHR}50   & \texttt{Zhang}1    & \texttt{Zhang}5 & \texttt{Zhang}10   \\ \hline
IBM                     & 0.008 & 0.010 & 0.003 & 0.008 & 0.004 & 0.000& 0.000& 0.002 & 0.100 & 0.040 & 0.010& 0.027 \\ \hline
\multicolumn{1}{l}{MCD} & 0.858 & 0.938 & 0.416 & 0.044 & 0.098 & 0.785 &0.207 & 0.424 & 0.153 & 0.552 & 0.301 & 0.344\\ \hline\hline
\end{tabular}}
\end{table}
From \Cref{app:tab1}, we find that all tests except for \texttt{HHHR}50 reject $H_0$ at 5\% level  for IBM stock, and  all except for \texttt{GK}3 fail to reject $H_0$ for MCD.  This finding is consistent with  \cite{zhang2016white} using 5-minute data and cubic B-splines basis functions. {   Here, we note  that the constructed price process could be contaminated by microstructure noise \citep{ait2009high} and we conjecture that their presence could compromise power due to diminished signal-to-noise ratio. We leave this for future research.}



\subsection{Mortality Data}
Human Mortality Database (\url{https://www.mortality.org/Home/Index}) has provided the scientific researchers a  high-quality harmonized mortality and population estimates.  For example, Figure \ref{fig:app} below gives a visual demonstration of the female (of age between 20 and 105) mortality  distribution time series in U.K. from 1970 to 2014. 

The yearly age-at-death distribution for a given country can be naturally viewed as a random element in the space of univariate distributions, which has been analyzed in many studies regarding modelling aspects for non-Euclidean valued random objects, see, e.g. \cite{petersen2019frechet},  \cite{dubeymuller2020}.
\begin{figure}[h]
\centering
\includegraphics[width=.78\textwidth]{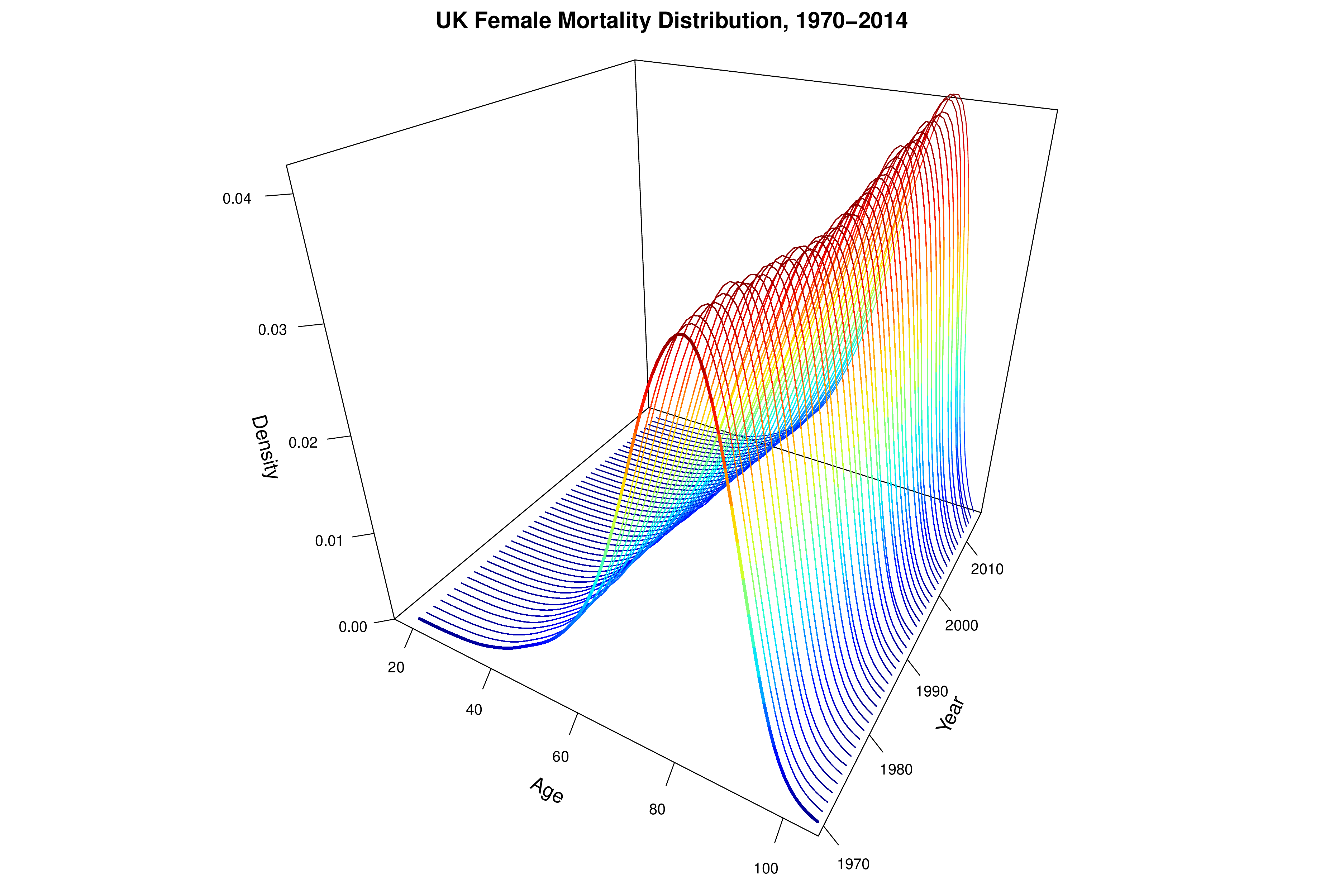}
\caption{UK female yearly age-at-death distribution, 1970-2014}\label{fig:app}
\end{figure}
However, many of them implicitly assume temporal independence for the sequence of random distributions, which could be doubtful. We  now formally test for this assumption.   

We analyze the yearly female mortality data in  27 European developed countries from 1970 to 2014. After applying our proposed CvM test and KS test with Wasserstein-1 and Wasserstein-2 metrics, we find that all  tests for all the countries reject $H_0$ at 5\% level. This provides  strong evidence of temporal dependence in mortality distribution time series. Therefore, more caution should be taken when it comes to the modelling of mortality data, as temporal dependence may not be negligible.


\section{Conclusion}\label{Sec:con}
In this paper, we have proposed a new method of testing serial independence for object-valued time series. It builds upon the distance covariance in metric spaces, and  captures all pairwise dependence. It is  fully nonparametric,  easy to implement with no tuning parameters, and is broadly applicable to many different data types as long as they live in a metric space of strong negative type. Our numerical results demonstrate accurate size for both permutation and bootstrap based tests and very competitive power performance even  in conventional Euclidean and Hilbert spaces.

This paper is concerned with serial dependence testing of object-valued time series, which is typically conducted before any specific model is fitted. {There is a recent  literature on modelling temporal dependence in object-valued  time series,  e.g. distributional autoregression in \cite{zhangJTSA2022}, \cite{zhu2021autoregressive}, \cite{ghodrati2023distributional}, and spherical autoregression in \cite{zhu2023spherical}, among others. }  It would be interesting  to investigate how to extend the test in this paper to  model diagnostic checking. This requires the researchers to properly define model residuals, which is a challenging task for object-valued data in metric spaces. {In addition, we assume the object-valued data are fully observed in the paper whereas in real applications they are  usually generated by pre-processing Euclidean-valued data. In the realm of Hilbert space, \cite{zhang2016sparse} have demonstrated that pre-smoothing in functional data analysis can exhibit a phase transition phenomenon, depending on the resolution of discretized observations.
It would also be helpful to investigate how pre-processing would impact  the behavior of our test.   Finally,  we assume the distance metric  $d(\cdot,\cdot)$ is   given a priori but its choice can play a significant role in data analysis; it usually depends on  the underlying inference problem and the characteristics of the dataset. How to choose the distance metric adaptively is very important yet also very challenging. 
We leave these topics for future research.  }

\clearpage

\begin{appendices}
\appendixpage
The supplementary material is organized as follows. Appendix \ref{Sec:AddSimul} contains additional simulation studies of {the Euclidean time seires}. Appendix \ref{Sec:lemma} provides some auxiliary lemmas and their proofs, and Appendix \ref{Sec:proof} gives all the proofs of theorems in the main text.


\section{Additional Numerical Studies}\label{Sec:AddSimul}


{In this section, we carry out some additional simulation studies based on the Euclidean time series. Appendix \ref{Sec:Simul-Uni} includes the simulation results of the univariate time series whereas Appendix \ref{Sec:Simul-Bi} contains those of the bivariate time series. In Appendix \ref{Sec:Simul-Multi}, we further investigate the finite sample performance of the proposed test statistics against different sample sizes $n$ and dimension $p$. Throughout this section, all the empirical rejection rates are average over $M=2500$ Monte Carlo replicates following the same computing procedures used in \Cref{Sec:simu}.}

\subsection{Univariate time series}\label{Sec:Simul-Uni}

Let $\Omega=\br$ and $d(x,y)=|x-y|$ for any $x,y\in\br$. We follow the setups of \citep{fokianos2017consistent} to generate a time series $\{Y_t\}_{t=1}^{n} \in \br$. Specifically, we set $n=200$ and consider four data generating processes (DGPs): (i) IID: $Y_t = \epsilon_t$; (ii) NMA(2): $Y_t = \epsilon_t \epsilon_{t-1} \epsilon_{t-2}$; (iii) ARCH(2): $Y_t = \sigma_t \epsilon_t$, $\sigma_t^2 = 0.5 + 0.8Y_{t-1}^2 + 0.1Y_{t-2}^2$; (iv) TAR(1):
    $ Y_t 
    = \left\{
        \begin{array}{ll}
            -1.5 Y_{t-1} + \epsilon_t, & Y_{t-1} < 0  \\
            0.5 Y_{t-1} + \epsilon_t,  & Y_{t-1} \geq 0
        \end{array}
      \right.$,
where $\{\epsilon_t\} \stsim{i.i.d.} \cn(0,1)$ in each DGP.

Apart from that of the proposed test statistic, we also report the performance of the following methods {with a wide range of parameters} for comparison:
\begin{enumerate}[label=(\roman*)]
    \item the classic Box–Pierce test (denoted by \texttt{BP}) and Ljung–Box test (denoted by \texttt{LB}). Both \texttt{BP} and \texttt{LB} can be implemented using \texttt{R} function \texttt{Box.test} and we consider the number of lags {$K\in\{1,3,6,9,12,15\}$};
    
    \item the consistent testing for serial correlation proposed by \cite{hong1996consistent} (denoted by \texttt{Hong}) with a pre-specified kernel $k(\cdot)$ and a bandwidth parameter $p_n$.
    {In particular, we use the Daniell kernel (D), the Parzen kernel (P), the Bartlett kernel (B), and the QS kernel (Q). As for the selection of $p_n$, it is suggested by \cite{hong1996consistent} to select $p_n\in\{\lrfl{\log(n)}, \lrfl{3*n^{0.2}}, \lrfl{3*n^{0.3}}\}$, which corresponds to $p_n\in\{5,9,15\}$ with $n=200$. Here, we consider a larger set of $p_n$ with $p_n\in\{3,6,9,12,15\}$.}
    
    \item the automatic Portmanteau test in \cite{escanciano2009automatic} (denoted by \texttt{EL}) with an upper bound $d \in \{25,50\}$.
    {In \cite{escanciano2009automatic}, the recommended value of $q$ is 2.4. We consider a larger selection of $q \in \{1.8, 2.1, 2.4, 2.7, 3.0\}$.}
    
    \item the consistent testing for pairwise dependence by \cite{fokianos2017consistent} (denoted by \texttt{FP}) with the same kernel function $k(\cdot)$ and bandwidth parameter $p_n$ as those used for \texttt{Hong}. {In the original paper, the suggested selection of $p_n$ include $\{\lrfl{n^{0.2}}, \lrfl{n^{0.4}}, \lrfl{n^{0.6}}\}$, which corresponds to $\{3,9,25\}$. In our setting, we consider $p_n\in\{3,9,15,21,25\}$.}
\end{enumerate}

{
Table \ref{Tab:Uni-R2} summarizes the empirical rejection rates of all the methods under each DGP. Here, \texttt{H}-D3 stands for the test statistic \texttt{Hong} with the Daniell kernel and $p_n=3$, \texttt{EL25-2.4} stands for the test statistic \texttt{EL} with $d=25$ and $q=2.4$. Similarly, the other concatenations of initials represent the corresponding parameter combinations.
}

\begin{table}[!h]
    \caption{Empirical rejection rate of univariate time series when $n=200$}
    \label{Tab:Uni-R2}
    \centering
    \scalebox{0.62}{
    \begin{tabular}{c|c|c|c|c||c|c|c|c||c|c|c|c||c|c|c|c}
    \hline\hline
        Method & \multicolumn{4}{c||}{IID} & \multicolumn{4}{c||}{NMA(2)} & \multicolumn{4}{c||}{ARCH(2)} & \multicolumn{4}{c}{TAR(2)} \\ \hline
        \multirow{2}{*}{Proposed} & CvM-P & 0.052 & CvM-B & 0.054 & CvM-P & 1.000 & CvM-B & 0.997 & CvM-P & 0.765 & CvM-B & 0.582 & CvM-P & 0.995 & CvM-B & 0.994 \\
        & KS-P & 0.055 & KS-B & 0.056 & KS-P & 1.000 & KS-B & 0.995 & KS-P & 0.734 & KS-B & 0.564 & KS-P & 0.993 & KS-B & 0.991 \\ \hline
        \multirow{3}{*}{\texttt{BP}} & \texttt{BP}-1 & 0.045 & \texttt{BP}-9 & 0.045 & \texttt{BP}-1 & 0.416 & \texttt{BP}-9 & 0.272 & \texttt{BP}-1 & 0.336 & \texttt{BP}-9 & 0.358 & \texttt{BP}-1 & 0.051 & \texttt{BP}-9 & 0.053 \\
        & \texttt{BP}-3 & 0.047 & \texttt{BP}-12 & 0.039 & \texttt{BP}-3 & 0.374 & \texttt{BP}-12 & 0.227 & \texttt{BP}-3 & 0.426 & \texttt{BP}-12 & 0.327 & \texttt{BP}-3 & 0.055 & \texttt{BP}-12 & 0.056 \\
        & \texttt{BP}-6 & 0.047 & \texttt{BP}-15 & 0.038 & \texttt{BP}-6 & 0.297 & \texttt{BP}-15 & 0.204 & \texttt{BP}-6 & 0.402 & \texttt{BP}-15 & 0.293 & \texttt{BP}-6 & 0.050 & \texttt{BP}-15 & 0.045 \\ \hline
        \multirow{3}{*}{\texttt{LB}} & \texttt{LB}-1 & 0.049 & \texttt{LB}-9 & 0.051 & \texttt{LB}-1 & 0.419 & \texttt{LB}-9 & 0.281 & \texttt{LB}-1 & 0.339 & \texttt{LB}-9 & 0.374 & \texttt{LB}-1 & 0.051 & \texttt{LB}-9 & 0.059 \\
        & \texttt{LB}-3 & 0.050 & \texttt{LB}-12 & 0.050 & \texttt{LB}-3 & 0.386 & \texttt{LB}-12 & 0.246 & \texttt{LB}-3 & 0.435 & \texttt{LB}-12 & 0.343 & \texttt{LB}-3 & 0.062 & \texttt{LB}-12 & 0.072 \\
        & \texttt{LB}-6 & 0.050 & \texttt{LB}-15 & 0.051 & \texttt{LB}-6 & 0.308 & \texttt{LB}-15 & 0.229 & \texttt{LB}-6 & 0.415 & \texttt{LB}-15 & 0.318 & \texttt{LB}-6 & 0.055 & \texttt{LB}-15 & 0.065 \\ \hline
        \multirow{10}{*}{\texttt{H}} & \texttt{H}-D3 & 0.068 & \texttt{H}-P3 & 0.067 & \texttt{H}-D3 & 0.463 & \texttt{H}-P3& 0.475 & \texttt{H}-D3 & 0.411 & \texttt{H}-P3 & 0.438 & \texttt{H}-D3 & 0.074 & \texttt{H}-P3 & 0.078 \\
        & \texttt{H}-D6 & 0.067 & \texttt{H}-P6 & 0.069 & \texttt{H}-D6 & 0.453 & \texttt{H}-P6 & 0.452 & \texttt{H}-D6 & 0.476 & \texttt{H}-P6 & 0.494 & \texttt{H}-D6 & 0.077 & \texttt{H}-P6 & 0.078 \\
        & \texttt{H}-D9 & 0.067 & \texttt{H}-P9 & 0.066 & \texttt{H}-D9 & 0.432 & \texttt{H}-P9 & 0.414 & \texttt{H}-D9 & 0.492 & \texttt{H}-P9 & 0.491 & \texttt{H}-D9 & 0.078 & \texttt{H}-P9 & 0.080 \\
        & \texttt{H}-D12 & 0.066 & \texttt{H}-P12 & 0.069 & \texttt{H}-D12 & 0.391 & \texttt{H}-P12 & 0.374 & \texttt{H}-D12 & 0.475 & \texttt{H}-P12 & 0.465 & \texttt{H}-D12 & 0.074 & \texttt{H}-P12 & 0.074 \\
        & \texttt{H}-D15 & 0.071 & \texttt{H}-P15 & 0.069 & \texttt{H}-D15 & 0.358 & \texttt{H}-P15 & 0.350 & \texttt{H}-D15 & 0.455 & \texttt{H}-P15 & 0.452 & \texttt{H}-D15 & 0.077 & \texttt{H}-P15 & 0.076 \\ \cline{2-17}
        & \texttt{H}-Q3 & 0.067 & \texttt{H}-B3 & 0.068 & \texttt{H}-Q3 & 0.470 & \texttt{H}-B3 & 0.468 & \texttt{H}-Q3 & 0.410 & \texttt{H}-B3 & 0.401 & \texttt{H}-Q3 & 0.076 & \texttt{H}-B3 & 0.077 \\
        & \texttt{H}-Q6 & 0.066 & \texttt{H}-B6 & 0.070 & \texttt{H}-Q6 & 0.456 & \texttt{H}-B6 & 0.467 & \texttt{H}-Q6 & 0.483 & \texttt{H}-B6 & 0.471 & \texttt{H}-Q6 & 0.078 & \texttt{H}-B6 & 0.078 \\
        & \texttt{H}-Q9 & 0.068 & \texttt{H}-B9 & 0.070 & \texttt{H}-Q9 & 0.433 & \texttt{H}-B9 & 0.448 & \texttt{H}-Q9 & 0.493 & \texttt{H}-B9 & 0.498 & \texttt{H}-Q9 & 0.075 & \texttt{H}-B9 & 0.078 \\
        & \texttt{H}-Q12 & 0.069 & \texttt{H}-B12 & 0.069 & \texttt{H}-Q12 & 0.388 & \texttt{H}-B12 & 0.420 & \texttt{H}-Q12 & 0.478 & \texttt{H}-B12 & 0.492 & \texttt{H}-Q12 & 0.076 & \texttt{H}-B12 & 0.079 \\
        & \texttt{H}-Q15 & 0.073 & \texttt{H}-B15 & 0.070 & \texttt{H}-Q15 & 0.364 & \texttt{H}-B15 & 0.389 & \texttt{H}-Q15 & 0.455 & \texttt{H}-B15 & 0.472 & \texttt{H}-Q15 & 0.080 & \texttt{H}-B15 & 0.079 \\ \hline
        \multirow{10}{*}{\texttt{FP}} & \texttt{FP}-D3 & 0.052 & \texttt{FP}-P3 & 0.054 & \texttt{FP}-D3 & 1.000 & \texttt{FP}-P3 & 1.000 & \texttt{FP}-D3 & 0.899 & \texttt{FP}-P3 & 0.905 & \texttt{FP}-D3 & 0.998 & \texttt{FP}-P3 & 0.998 \\
        & \texttt{FP}-D9 & 0.042 & \texttt{FP}-P9 & 0.042 & \texttt{FP}-D9 & 0.966 & \texttt{FP}-P9 & 0.970 & \texttt{FP}-D9 & 0.793 & \texttt{FP}-P9 & 0.794 & \texttt{FP}-D9 & 0.938 & \texttt{FP}-P9 & 0.931 \\
        & \texttt{FP}-D15 & 0.035 & \texttt{FP}-P15 & 0.030 & \texttt{FP}-D15 & 0.729 & \texttt{FP}-P15 & 0.790 & \texttt{FP}-D15 & 0.618 & \texttt{FP}-P15 & 0.654 & \texttt{FP}-D15 & 0.698 & \texttt{FP}-P15 & 0.752 \\
        & \texttt{FP}-D21 & 0.036 & \texttt{FP}-P21 & 0.030 & \texttt{FP}-D21 & 0.347 & \texttt{FP}-P21 & 0.521 & \texttt{FP}-D21 & 0.428 & \texttt{FP}-P21 & 0.518 & \texttt{FP}-D21 & 0.422 & \texttt{FP}-P21 & 0.564 \\
        & \texttt{FP}-D25 & 0.034 & \texttt{FP}-P25 & 0.030 & \texttt{FP}-D25 & 0.278 & \texttt{FP}-P25 & 0.390 & \texttt{FP}-D25 & 0.396 & \texttt{FP}-P25 & 0.443 & \texttt{FP}-D25 & 0.355 & \texttt{FP}-P25 & 0.446 \\ \cline{2-17}
        & \texttt{FP}-Q3 & 0.054 & \texttt{FP}-B3 & 0.055 & \texttt{FP}-Q3 & 1.000 & \texttt{FP}-B3 & 1.000 & \texttt{FP}-Q3 & 0.904 & \texttt{FP}-B3 & 0.904 & \texttt{FP}-Q3 & 0.999 & \texttt{FP}-B3 & 0.999 \\
        & \texttt{FP}-Q9 & 0.046 & \texttt{FP}-B9 & 0.046 & \texttt{FP}-Q9 & 0.981 & \texttt{FP}-B9 & 0.991 & \texttt{FP}-Q9 & 0.816 & \texttt{FP}-B9 & 0.839 & \texttt{FP}-Q9 & 0.946 & \texttt{FP}-B9 & 0.972 \\
        & \texttt{FP}-Q15 & 0.037 & \texttt{FP}-B15 & 0.036 & \texttt{FP}-Q15 & 0.850 & \texttt{FP}-B15 & 0.935 & \texttt{FP}-Q15 & 0.690 & \texttt{FP}-B15 & 0.753 & \texttt{FP}-Q15 & 0.809 & \texttt{FP}-B15 & 0.906 \\
        & \texttt{FP}-Q21 & 0.033 & \texttt{FP}-B21 & 0.029 & \texttt{FP}-Q21 & 0.606 & \texttt{FP}-B21 & 0.806 & \texttt{FP}-Q21 & 0.557 & \texttt{FP}-B21 & 0.653 & \texttt{FP}-Q21 & 0.618 & \texttt{FP}-B21 & 0.772 \\
        & \texttt{FP}-Q25 & 0.033 & \texttt{FP}-B25 & 0.032 & \texttt{FP}-Q25 & 0.474 & \texttt{FP}-B25 & 0.690 & \texttt{FP}-Q25 & 0.482 & \texttt{FP}-B25 & 0.578 & \texttt{FP}-Q25 & 0.515 & \texttt{FP}-B25 & 0.682 \\ \hline
        \multirow{7}{*}{\texttt{EL}} & \texttt{EL}25-1.2 & 0.163 & \texttt{EL}-50-1.2& 0.209 & \texttt{EL}25-1.2 & 0.127 & \texttt{EL}50-1.2 & 0.170 & \texttt{EL}25-1.2 & 0.111 & \texttt{EL}50-1.2 & 0.163 & \texttt{EL}25-1.2 & 0.180 & \texttt{EL}50-1.2 & 0.240 \\
        & \texttt{EL}25-1.5 & 0.120 & \texttt{EL}50-1.5 & 0.152 & \texttt{EL}25-1.5 & 0.080 & \texttt{EL}50-1.5 & 0.104 & \texttt{EL}25-1.5 & 0.072 & \texttt{EL}50-1.5 & 0.108 & \texttt{EL}25-1.5 & 0.124 & \texttt{EL}50-1.5 & 0.168 \\
        & \texttt{EL}25-1.8 & 0.090 & \texttt{EL}50-1.8 & 0.113 & \texttt{EL}25-1.8 & 0.070 & \texttt{EL}50-1.8 & 0.078 & \texttt{EL}25-1.8 & 0.054 & \texttt{EL}50-1.8 & 0.066 & \texttt{EL}25-1.8 & 0.093 & \texttt{EL}50-1.8 & 0.119 \\
        & \texttt{EL}25-2.1 & 0.078 & \texttt{EL}50-2.1 & 0.090 & \texttt{EL}25-2.1 & 0.066 & \texttt{EL}50-2.1 & 0.068 & \texttt{EL}25-2.1 & 0.045 & \texttt{EL}50-2.1 & 0.051 & \texttt{EL}25-2.1 & 0.074 & \texttt{EL}50-2.1 & 0.088 \\
        & \texttt{EL}25-2.4 & 0.070 & \texttt{EL}50-2.4 & 0.076 & \texttt{EL}25-2.4 & 0.064 & \texttt{EL}50-2.4 & 0.064 & \texttt{EL}25-2.4 & 0.043 & \texttt{EL}50-2.4 & 0.046 & \texttt{EL}25-2.4 & 0.066 & \texttt{EL}50-2.4 & 0.073 \\
        & \texttt{EL}25-2.7 & 0.069 & \texttt{EL}50-2.7 & 0.071 & \texttt{EL}25-2.7 & 0.063 & \texttt{EL}50-2.7 & 0.063 & \texttt{EL}25-2.7 & 0.041 & \texttt{EL}50-2.7 & 0.042 & \texttt{EL}25-2.7 & 0.063 & \texttt{EL}50-2.7 & 0.067 \\
        & \texttt{EL}25-3.0 & 0.068 & \texttt{EL}50-3.0 & 0.069 & \texttt{EL}25-3.0 & 0.063 & \texttt{EL}50-3.0 & 0.063 & \texttt{EL}25-3.0 & 0.040 & \texttt{EL}50-3.0 & 0.041 & \texttt{EL}25-3.0 & 0.063 & \texttt{EL}50-3.0 & 0.065 \\
    \hline\hline
    \end{tabular}}
\end{table}

{Under the null, our method together with  \texttt{LB} both achieve accurate size. For \texttt{BP} and \texttt{FP}, the size appears conservative when the truncation number/bandwidth becomes large. 
We also observe some degree of over-rejection with \texttt{H} and \texttt{EL} for the range of tuning parameters under examination. The power of our test against the alternative depends on the specific DGP. In particular, the proposed test is among the most powerful tests against the TAR(2) model. Note that both NMA(2) and ARCH(2) models have serial dependence but no serial correlation, and our method has very high power against NMA(2) model but exhibits slight power loss against the ARCH(2) model when comparing to the most powerful test in this case (i.e., \texttt{FP} with $p_n=3$). This observation is consistent with that in \Cref{Sec:Simul-Func}.

As shown from the simulation results, the performance of all the competing methods depends on the choice of tuning parameters. It is still possible that our test is outperformed by one of these  methods with carefully selected tuning parameters under specific DGP. For example, \texttt{FP} achieves a very accurate size and the highest power but only when specific parameters are chosen. Although empirical recommendations of tuning parameter selection are provided for each comparison method, there is no theoretical or practical guarantee, making  parameter selection difficult in practice. In comparison, the proposed test is  tuning-free and is convenient to implement, which is a natural advantage in real-world applications.}

\subsection{Bivariate time series}\label{Sec:Simul-Bi}

Consider $\Omega=\br^2$ with the Euclidean distance $d(x,y)=|x-y|_2$, and we generate $\{\epsilon_t\} \stsim{i.i.d.} \cn_2(0,\Sigma)$ with $\Sigma = \mat{1}{\rho}{\rho}{1}$ and $\rho\in\{0, 0.4, 0.7\}$. We fix $n=200$ and generate $\{Y_t\}_{t=1}^{n}$ by: (i) IID: $Y_t = \epsilon_t$; (ii) NMA(2): $Y_{t,i} = \epsilon_{t,i}\epsilon_{t-1,i}\epsilon_{t-2,i}$; (iii) ARCH(2): $Y_{t,i} = h_{t,i}^{1/2}\epsilon_{t,i}$, where
\begin{equation*}
    \cvec{h_{t,1}}{h_{t,2}} 
  = \cvec{0.003}{0.005} 
  + \mat{0.2}{0.1}{0.1}{0.3} \cvec{Y_{t-1,1}^2}{Y_{t-1,2}^2} 
  + \mat{0.4}{0.05}{0.05}{0.5} \cvec{h_{t-1,1}}{h_{t-1,2}}.
\end{equation*}
and (iv) MAR(2): $Y_t = \mat{0.04}{-0.1}{0.11}{0.5}Y_{t-1}+\epsilon_t$.

The multivariate Ljung–Box test statistic \texttt{mLB} proposed in \cite{hosking1980multivariate} is used for comparison and we select the number of lags $h$ from $\{1,3,6,9\}$. Additionally, we apply the multivariate testings \texttt{mADCV} and \texttt{mADCF} in \cite{fokianos2018testing} using the same kernel functions and bandwidth parameters as mentioned in \Cref{Sec:Simul-Cov}.

\begin{table}[!h]
    \caption{Empirical rejection rate of bivariate time series when $n=200$}
    \label{Tab:Bi-2}
    \centering
    \scalebox{0.65}{
    \begin{tabular}{ccc|ccc|ccc|ccc|ccc}
    \hline\hline
        \multicolumn{3}{c|}{\multirow{2}{*}{Method}} & \multicolumn{3}{c|}{IID} & \multicolumn{3}{c|}{NMA(2)} & \multicolumn{3}{c|}{ARCH(2)} & \multicolumn{3}{c}{MAR(2)} \\ \cline{4-15}
        & & & $\rho=0$ & $\rho=0.4$ & $\rho=0.7$ & $\rho=0$ & $\rho=0.4$ & $\rho=0.7$ & $\rho=0$ & $\rho=0.4$ & $\rho=0.7$ & $\rho=0$ & $\rho=0.4$ & $\rho=0.7$ \\ \hline
        \multicolumn{1}{c|}{\multirow{4}{*}{Proposed}} & \multicolumn{1}{c|}{\multirow{2}{*}{\textsc{CvM}}} & Permt & 0.053 & 0.052 & 0.056 & 0.997 & 0.999 & 1.000 & 0.348 & 0.346 & 0.348 & 1.000 & 1.000 & 1.000 \\ \cline{3-15}
        \multicolumn{1}{c|}{} & \multicolumn{1}{c|}{} & Boot & 0.057 & 0.053 & 0.056 & 0.978 & 0.989 & 1.000 & 0.261 & 0.240 & 0.242 & 1.000 & 1.000 & 1.000 \\ \cline{2-15}
        \multicolumn{1}{c|}{} & \multicolumn{1}{c|}{\multirow{2}{*}{\textsc{KS}}} & Permt & 0.051 & 0.053 & 0.056 & 0.996 & 0.997 & 1.000 & 0.342 & 0.319 & 0.335 & 1.000 & 1.000 & 1.000 \\ \cline{3-15}
        \multicolumn{1}{c|}{} & \multicolumn{1}{c|}{} & Boot & 0.051 & 0.058 & 0.060 & 0.971 & 0.984 & 0.997 & 0.224 & 0.230 & 0.234 & 1.000 & 1.000 & 1.000 \\ \hline
        \multicolumn{2}{c|}{\multirow{4}{*}{\texttt{mLB}}} & $h=1$ & 0.115 & 0.034 & 0.016 & 0.522 & 0.470 & 0.395 & 0.139 & 0.066 & 0.042 & 0.157 & 0.020 & 0.012 \\ \cline{3-15}
        \multicolumn{2}{c|}{} & $h=3$ & 0.172 & 0.025 & 0.004 & 0.522 & 0.468 & 0.338 & 0.167 & 0.046 & 0.035 & 0.160 & 0.020 & 0.005 \\ \cline{3-15}
        \multicolumn{2}{c|}{} & $h=6$ & 0.194 & 0.009 & 0.001 & 0.467 & 0.417 & 0.270 & 0.130 & 0.031 & 0.017 & 0.150 & 0.008 & 0.000 \\ \cline{3-15}
        \multicolumn{2}{c|}{} & $h=9$ & 0.213 & 0.007 & 0.001 & 0.442 & 0.403 & 0.229 & 0.106 & 0.021 & 0.009 & 0.146 & 0.004 & 0.000 \\ \hline
        \multicolumn{1}{c|}{\multirow{12}{*}{\texttt{mADCV}}} & \multicolumn{1}{c|}{\multirow{3}{*}{TC}} & $\lambda=1/10$ & 0.047 & 0.038 & 0.038 & 0.779 & 0.756 & 0.682 & 0.242 & 0.256 & 0.259 & 0.997 & 1.000 & 1.000 \\ \cline{3-15}
        \multicolumn{1}{c|}{} & \multicolumn{1}{c|}{} & $\lambda=2/10$ & 0.038 & 0.044 & 0.054 & 0.732 & 0.714 & 0.637 & 0.244 & 0.250 & 0.255 & 0.991 & 0.999 & 1.000 \\ \cline{3-15}
        \multicolumn{1}{c|}{} & \multicolumn{1}{c|}{} & $\lambda=3/10$ & 0.056 & 0.049 & 0.043 & 0.672 & 0.647 & 0.594 & 0.223 & 0.239 & 0.236 & 0.972 & 0.995 & 0.999 \\ \cline{2-15}
        \multicolumn{1}{c|}{} & \multicolumn{1}{c|}{\multirow{3}{*}{DAN}} & $\lambda=1/10$ & 0.047 & 0.045 & 0.048 & 0.869 & 0.846 & 0.762 & 0.233 & 0.237 & 0.250 & 1.000 & 1.000 & 1.000 \\ \cline{3-15}
        \multicolumn{1}{c|}{} & \multicolumn{1}{c|}{} & $\lambda=2/10$ & 0.056 & 0.041 & 0.049 & 0.873 & 0.833 & 0.768 & 0.271 & 0.268 & 0.279 & 1.000 & 1.000 & 1.000 \\ \cline{3-15}
        \multicolumn{1}{c|}{} & \multicolumn{1}{c|}{} & $\lambda=3/10$ & 0.047 & 0.046 & 0.051 & 0.844 & 0.811 & 0.742 & 0.272 & 0.285 & 0.291 & 0.999 & 1.000 & 1.000 \\ \cline{2-15}
        \multicolumn{1}{c|}{} & \multicolumn{1}{c|}{\multirow{3}{*}{PAR}} & $\lambda=1/10$ & 0.051 & 0.043 & 0.042 & 0.866 & 0.831 & 0.755 & 0.239 & 0.256 & 0.245 & 1.000 & 1.000 & 1.000 \\ \cline{3-15}
        \multicolumn{1}{c|}{} & \multicolumn{1}{c|}{} & $\lambda=2/10$ & 0.050 & 0.038 & 0.045 & 0.855 & 0.822 & 0.744 & 0.269 & 0.274 & 0.291 & 0.999 & 1.000 & 1.000 \\ \cline{3-15}
        \multicolumn{1}{c|}{} & \multicolumn{1}{c|}{} & $\lambda=3/10$ & 0.055 & 0.045 & 0.044 & 0.825 & 0.800 & 0.721 & 0.269 & 0.280 & 0.279 & 0.999 & 1.000 & 1.000 \\ \cline{2-15}
        \multicolumn{1}{c|}{} & \multicolumn{1}{c|}{\multirow{3}{*}{BAR}} & $\lambda=1/10$ & 0.051 & 0.044 & 0.049 & 0.871 & 0.839 & 0.763 & 0.220 & 0.225 & 0.226 & 1.000 & 1.000 & 1.000 \\ \cline{3-15}
        \multicolumn{1}{c|}{} & \multicolumn{1}{c|}{} & $\lambda=2/10$ & 0.052 & 0.039 & 0.037 & 0.869 & 0.839 & 0.758 & 0.255 & 0.262 & 0.274 & 1.000 & 1.000 & 1.000 \\ \cline{3-15}
        \multicolumn{1}{c|}{} & \multicolumn{1}{c|}{} & $\lambda=3/10$ & 0.053 & 0.046 & 0.049 & 0.859 & 0.830 & 0.751 & 0.274 & 0.278 & 0.288 & 0.999 & 1.000 & 1.000 \\ \hline
        \multicolumn{1}{c|}{\multirow{12}{*}{\texttt{mADCF}}} & \multicolumn{1}{c|}{\multirow{3}{*}{TC}} & $\lambda=1/10$ & 0.046 & 0.043 & 0.045 & 0.782 & 0.746 & 0.677 & 0.232 & 0.260 & 0.258 & 0.997 & 1.000 & 1.000 \\ \cline{3-15}
        \multicolumn{1}{c|}{} & \multicolumn{1}{c|}{} & $\lambda=2/10$ & 0.045 & 0.042 & 0.053 & 0.710 & 0.694 & 0.626 & 0.226 & 0.256 & 0.250 & 0.987 & 1.000 & 0.999 \\ \cline{3-15}
        \multicolumn{1}{c|}{} & \multicolumn{1}{c|}{} & $\lambda=3/10$ & 0.064 & 0.046 & 0.046 & 0.639 & 0.638 & 0.577 & 0.222 & 0.235 & 0.236 & 0.964 & 0.993 & 0.999 \\ \cline{2-15}
        \multicolumn{1}{c|}{} & \multicolumn{1}{c|}{\multirow{3}{*}{DAN}} & $\lambda=1/10$ & 0.053 & 0.042 & 0.045 & 0.873 & 0.845 & 0.749 & 0.226 & 0.217 & 0.239 & 1.000 & 1.000 & 1.000 \\ \cline{3-15}
        \multicolumn{1}{c|}{} & \multicolumn{1}{c|}{} & $\lambda=2/10$ & 0.055 & 0.036 & 0.046 & 0.863 & 0.823 & 0.744 & 0.265 & 0.260 & 0.266 & 1.000 & 1.000 & 1.000 \\ \cline{3-15}
        \multicolumn{1}{c|}{} & \multicolumn{1}{c|}{} & $\lambda=3/10$ & 0.058 & 0.044 & 0.048 & 0.832 & 0.816 & 0.727 & 0.264 & 0.279 & 0.288 & 0.998 & 1.000 & 1.000 \\ \cline{2-15}
        \multicolumn{1}{c|}{} & \multicolumn{1}{c|}{\multirow{3}{*}{PAR}} & $\lambda=1/10$ & 0.052 & 0.037 & 0.044 & 0.862 & 0.841 & 0.760 & 0.235 & 0.233 & 0.250 & 1.000 & 1.000 & 1.000 \\ \cline{3-15}
        \multicolumn{1}{c|}{} & \multicolumn{1}{c|}{} & $\lambda=2/10$ & 0.058 & 0.041 & 0.043 & 0.860 & 0.830 & 0.750 & 0.257 & 0.266 & 0.282 & 0.999 & 1.000 & 1.000 \\ \cline{3-15}
        \multicolumn{1}{c|}{} & \multicolumn{1}{c|}{} & $\lambda=3/10$ & 0.057 & 0.045 & 0.053 & 0.825 & 0.800 & 0.715 & 0.262 & 0.277 & 0.286 & 0.998 & 1.000 & 1.000 \\ \cline{2-15}
        \multicolumn{1}{c|}{} & \multicolumn{1}{c|}{\multirow{3}{*}{BAR}} & $\lambda=1/10$ & 0.052 & 0.046 & 0.051 & 0.868 & 0.839 & 0.753 & 0.217 & 0.211 & 0.224 & 1.000 & 1.000 & 1.000 \\ \cline{3-15}
        \multicolumn{1}{c|}{} & \multicolumn{1}{c|}{} & $\lambda=2/10$ & 0.053 & 0.042 & 0.049 & 0.868 & 0.837 & 0.757 & 0.252 & 0.249 & 0.273 & 1.000 & 1.000 & 1.000 \\ \cline{3-15}
        \multicolumn{1}{c|}{} & \multicolumn{1}{c|}{} & $\lambda=3/10$ & 0.054 & 0.043 & 0.050 & 0.860 & 0.829 & 0.736 & 0.268 & 0.272 & 0.278 & 0.999 & 1.000 & 1.000 \\
    \hline\hline
    \end{tabular}}
\end{table}

All the simulation results are reported in Table \ref{Tab:Bi-2}. Similar to the observations in Appendix \ref{Sec:Simul-Uni}, \textsc{KS} and \textsc{CvM} have comparable size accuracy and power across the table, 
{and for both types of statistics, \texttt{Permt} yields higher power than \texttt{Boot} on the serially-uncorrelated time series NMA(2) and ARCH(2).} Additionally, both test statistics seem robust to the componentwise dependence within the data, implying that the proposed test can handle various dependence structures. 

By contrast,  \texttt{mLB} exhibits noticable size distortion and substantial power loss.
As for \texttt{mADCV} and \texttt{mADCF}, both tests are accurate in size but have slight power loss when compared to our proposed test statistic, {especially when using the truncated kernel. The choice of the bandwidth also has some impact on the finite sample performance of \texttt{mADCV} and \texttt{mADCF}. In particular, both test statistics have higher power with a smaller $\lambda$ under NMA(2) whereas with a larger $\lambda$ under ARCH(2). In addition, the strengthening componentwise dependence of the NMA(2) model leads to a more severe power loss of \texttt{mADCV} and \texttt{mADCF}.}

\subsection{Multivariate time series}\label{Sec:Simul-Multi}

{Lastly, we mimic Table 1 in \cite{zhou2012measuring} to investigate the impact of $n$ and $p$ on the power of our test. To this end, We generate a $p$-dimensinoal time series from a VAR(1) model, i.e.
\begin{equation*}
    Y_t = \rho I_p Y_{t-1} + \varepsilon_t \in \br^p, \qquad t=1,\cdots,n,
\end{equation*}
where $\{\varepsilon_t\}_{t=1}^{n} \stsim{iid} \cn_p(0,I_p)$ is generated from the standard normal distribution in $\br^p$. Here, we consider $n\in\{100,250\}$, $p\in\{2,5,10,20,30\}$ and we fix $\rho=0.2$. The comparison methods are the same as in Appendix \ref{Sec:Simul-Bi}.

According to the numerical results in \Cref{Tab:Bi-R1}, the power of all the methods increases significantly as $n$ increases. As the dimension $p$ increases, the rejection rates of the proposed methods and \texttt{mLB} both increase. In particular, when $p$ is small, the proposed test could be outperformed by \texttt{mLB} using carefully selected parameters. However, under the moderate-dimensional setting, our method has the highest power, regardless of the sample size.

As $p$ increases, we observe noticeable power loss of \texttt{mADCV} and \texttt{mADCF}, which is consistent with our observation in \Cref{Sec:Simul-Cov}. We conjecture this is due to deteriorated long-run covariance estimation when dimension  increases. This is a new finding as the simulation studies in \cite{fokianos2018testing} are limited to the two-dimensional case. }


\begin{table}[!h]
    \caption{Empirical rejection rate of multivariate time series with different $n$ and $p$}
    \label{Tab:Bi-R1}
    \centering
    \scalebox{0.8}{
    \begin{tabular}{c|c|c||ccccc|ccccc}
    \hline\hline
        \multicolumn{3}{c||}{\multirow{2}{*}{Method}} & \multicolumn{5}{c|}{$n=100$} & \multicolumn{5}{c}{$n=250$} \\ \cline{4-13}
        \multicolumn{3}{c||}{} & $p=2$ & $p=5$ & $p=10$ & $p=20$ & $p=30$ & $p=2$ & $p=5$ & $p=10$ & $p=20$ & $p=30$ \\ \hline
        \multirow{4}{*}{Proposed} & \multirow{2}{*}{\textsc{CvM}} & Permt & 0.488 & 0.606 & 0.749 & 0.899 & 0.970 & 0.929 & 0.994 & 0.998 & 1.000 & 1.000 \\
        & & Boot & 0.485 & 0.604 & 0.755 & 0.898 & 0.970 & 0.932 & 0.992 & 0.998 & 1.000 & 1.000 \\ \cline{2-13}
        & \multirow{2}{*}{\textsc{KS}} & Permt & 0.465 & 0.572 & 0.705 & 0.863 & 0.968 & 0.911 & 0.988 & 0.998 & 0.999 & 1.000 \\
        & & Boot & 0.465 & 0.579 & 0.708 & 0.862 & 0.966 & 0.912 & 0.990 & 0.998 & 0.999 & 1.000 \\ \hline
        \multicolumn{2}{c|}{\multirow{4}{*}{\texttt{mLB}}} & $h=1$ & 0.616 & 0.729 & 0.773 & 0.835 & 0.872 & 0.930 & 0.966 & 0.977 & 0.970 & 0.978 \\
        \multicolumn{2}{c|}{} & $h=3$ & 0.523 & 0.674 & 0.726 & 0.823 & 0.881 & 0.858 & 0.926 & 0.955 & 0.953 & 0.966 \\
        \multicolumn{2}{c|}{} & $h=6$ & 0.460 & 0.629 & 0.694 & 0.828 & 0.893 & 0.782 & 0.875 & 0.920 & 0.931 & 0.940 \\
        \multicolumn{2}{c|}{} & $h=9$ & 0.435 & 0.601 & 0.689 & 0.824 & 0.906 & 0.737 & 0.846 & 0.890 & 0.908 & 0.930 \\ \hline
        \multirow{12}{*}{\texttt{mADCV}} & \multirow{3}{*}{TC} & $\lambda=1/10$ & 0.214 & 0.105 & 0.010 & 0.000 & 0.000 & 0.620 & 0.639 & 0.462 & 0.088 & 0.007 \\
        & & $\lambda=2/10$ & 0.192 & 0.105 & 0.025 & 0.000 & 0.000 & 0.493 & 0.550 & 0.428 & 0.148 & 0.038 \\
        & & $\lambda=3/10$ & 0.191 & 0.128 & 0.035 & 0.002 & 0.000 & 0.433 & 0.497 & 0.462 & 0.274 & 0.137 \\ \cline{2-13}
        & \multirow{3}{*}{DAN} & $\lambda=1/10$ & 0.360 & 0.190 & 0.024 & 0.000 & 0.000 & 0.858 & 0.912 & 0.809 & 0.331 & 0.037 \\
        & & $\lambda=2/10$ & 0.305 & 0.171 & 0.031 & 0.000 & 0.000 & 0.760 & 0.814 & 0.722 & 0.343 & 0.095 \\
        & & $\lambda=3/10$ & 0.269 & 0.186 & 0.047 & 0.001 & 0.000 & 0.666 & 0.730 & 0.657 & 0.358 & 0.134 \\ \cline{2-13}
        & \multirow{3}{*}{PAR} & $\lambda=1/10$ & 0.346 & 0.180 & 0.021 & 0.000 & 0.000 & 0.833 & 0.893 & 0.768 & 0.270 & 0.024 \\
        & & $\lambda=2/10$ & 0.282 & 0.158 & 0.027 & 0.000 & 0.000 & 0.739 & 0.786 & 0.667 & 0.242 & 0.051 \\
        & & $\lambda=3/10$ & 0.257 & 0.172 & 0.041 & 0.000 & 0.000 & 0.636 & 0.711 & 0.629 & 0.322 & 0.119 \\ \cline{2-13}
        & \multirow{3}{*}{BAR} & $\lambda=1/10$ & 0.362 & 0.191 & 0.023 & 0.000 & 0.000 & 0.865 & 0.923 & 0.825 & 0.336 & 0.039 \\
        & & $\lambda=2/10$ & 0.318 & 0.176 & 0.029 & 0.000 & 0.000 & 0.794 & 0.859 & 0.766 & 0.319 & 0.061 \\
        & & $\lambda=3/10$ & 0.296 & 0.182 & 0.042 & 0.000 & 0.000 & 0.728 & 0.794 & 0.707 & 0.370 & 0.118 \\ \hline
        \multirow{12}{*}{\texttt{mADCF}} & \multirow{3}{*}{TC} & $\lambda=1/10$ & 0.208 & 0.106 & 0.012 & 0.000 & 0.000 & 0.614 & 0.630 & 0.439 & 0.074 & 0.007 \\
        & & $\lambda=2/10$ & 0.190 & 0.104 & 0.021 & 0.000 & 0.000 & 0.477 & 0.533 & 0.410 & 0.129 & 0.033 \\
        & & $\lambda=3/10$ & 0.195 & 0.125 & 0.034 & 0.002 & 0.000 & 0.422 & 0.494 & 0.483 & 0.294 & 0.164 \\ \cline{2-13}
        & \multirow{3}{*}{DAN} & $\lambda=1/10$ & 0.362 & 0.186 & 0.026 & 0.000 & 0.000 & 0.852 & 0.915 & 0.819 & 0.358 & 0.054 \\
        & & $\lambda=2/10$ & 0.293 & 0.163 & 0.030 & 0.000 & 0.000 & 0.763 & 0.820 & 0.721 & 0.337 & 0.099 \\
        & & $\lambda=3/10$ & 0.264 & 0.155 & 0.030 & 0.001 & 0.000 & 0.657 & 0.730 & 0.633 & 0.328 & 0.136 \\ \cline{2-13}
        & \multirow{3}{*}{PAR} & $\lambda=1/10$ & 0.338 & 0.174 & 0.021 & 0.000 & 0.000 & 0.832 & 0.886 & 0.772 & 0.269 & 0.033 \\
        & & $\lambda=2/10$ & 0.280 & 0.152 & 0.024 & 0.000 & 0.000 & 0.736 & 0.791 & 0.672 & 0.285 & 0.069 \\
        & & $\lambda=3/10$ & 0.252 & 0.155 & 0.034 & 0.001 & 0.000 & 0.623 & 0.705 & 0.618 & 0.332 & 0.132 \\ \cline{2-13}
        & \multirow{3}{*}{BAR} & $\lambda=1/10$ & 0.357 & 0.177 & 0.014 & 0.000 & 0.000 & 0.867 & 0.923 & 0.815 & 0.304 & 0.031 \\
        & & $\lambda=2/10$ & 0.317 & 0.167 & 0.023 & 0.000 & 0.000 & 0.794 & 0.847 & 0.730 & 0.274 & 0.054 \\
        & & $\lambda=3/10$ & 0.283 & 0.160 & 0.028 & 0.001 & 0.000 & 0.724 & 0.774 & 0.666 & 0.326 & 0.102 \\
    \hline\hline    
    \end{tabular}}    
\end{table}





\section{Auxiliary Lemmas}\label{Sec:lemma}

\begin{lem}\label{lem_hoeff}
For the  kernel $h$ in \eqref{kernel}, let $z_i=(x_i,y_i)$ and $Z_i=(X_i,Y_i)$ i.i.d., $i=1,\cdots,4$ and $$
h_c\left(z_1, \cdots, z_c\right)=\mathbb{E}\left\{h\left(z_1, \cdots, z_c, Z_{c+1}, \cdots, Z_4\right)\right\},
$$ 
where $c=1,2,3,4$, and define $\mathbb{E}h(Z_1,\cdots,Z_4)=\theta$.  Then,
\begin{flalign*}
h_1(z)=&\frac{1}{2}\left\{\mathbb{E} [d_{\nu}(x,X)d_{\nu}(y,Y)]+\theta\right\},\\
h_2(z_1,z_2)=&\frac{1}{6}\big\{d_{\nu}(x_1,x_2)d_{\nu}(y_1,y_2)+2\mathbb{E}[d_{\nu}(x_1,X)d_{\nu}(y_1,Y)]+2\mathbb{E}[d_{\nu}(x_2,X)d_{\nu}(y_2,Y)]\\&+\theta-\mathbb{E}[d_{\nu}(x_1,X)d_{\nu}(y_2,Y)]-\mathbb{E}[d_{\nu}(x_2,X)d_{\nu}(y_1,Y)]\big\},\\
h_3(z_1,z_2,z_3)=&\frac{1}{12}\big\{2\sum_{1\leq i_1<i_2\leq 3}d_{\nu}(x_{i_1},x_{i_2})d_{\nu}(y_{i_1},y_{i_2})+2 \sum_{1\leq i\leq 3}\mathbb{E}[d_{\nu}(x_i,X)d_{\nu}(y_i,Y)]\\&-\sum_{(i_1,i_2,i_3)}^{(1,2,3)}d_{\nu}(x_{i_1},x_{i_2})d_{\nu}(y_{i_1},y_{i_3})-\sum_{1\leq i_1\neq i_2\leq 3}\mathbb{E}[d_{\nu}(x_{i_1},X)d_{\nu}(y_{i_2},Y)]
\big\},\\
h_4(z_1,z_2,z_3,z_4) =&\frac{1}{4 !} \sum_{(i_1,i_2,i_3,i_4)}^{(1,2,3,4)}d_{\nu}(x_{i_1},x_{i_2})\left[d_{\nu}(y_{i_3},y_{i_4})+d_{\nu}(y_{i_1},y_{i_2})-2d_{\nu}(y_{i_1},y_{i_3})\right].
\end{flalign*}

Furthermore, let  $h^{(1)}(z)=h_1(z)-\theta$ and $h^{(c)}=h_c(z_1,\cdots,z_c)-\sum_{j=1}^{c-1}\sum_{(c,j)}h^{(j)}(z_{i_1},\cdots,z_{i_j})-\theta$, where   ${\sum}_{(c,j)}$ denotes the summation over all $j$-subsets $\{({i_1},\cdots,{i_j})\}$ of $\{1,\cdots,c\}$.  Then, 
\begin{equation}\label{hoeffding}
    U_n(h)=\sum_{c=0}^4{4\choose c}U_{n}(h^{(c)}),
\end{equation}
where $U_{n}(h^{(c)})$ is the $U$ statistic based on $h^{(c)}$.  

\end{lem}
\begin{proof}
See Section 9 in \cite{han2021generalized} or Section 1.2 of the supplement of \cite{zhang2018conditional}.  \eqref{hoeffding} follows directly by Hoeffding decomposition.
\end{proof}

\begin{lem}\label{lem_expand}
Denote by $\{\lambda_{\ell}\}_{\ell=1}^{\infty}$, and $\{e_{\ell}\}_{\ell=1}^{\infty}$ as the eigenvalues and orthonormal eigenfunctions corresponding to \eqref{eq_eigen}, then under Assumptions \ref{ass_moment}, 
\begin{equation}\label{K_kernel}
    \mathcal{K}(z,z'):=d_{\nu}(x,x')d_{\nu}(y,y')=\sum_{\ell=1}^{\infty}\lambda_{\ell} e_{\ell}(z)e_{\ell}(z').
\end{equation}
where  the series converges absolutely and uniformly on $(z,z') \in \Omega^2\times \Omega^2$.
\end{lem}

\begin{proof}
When $(\Omega,d)$ is of strong negative type, by Lemma 12 and Proposition 18 in \cite{sejdinovic2013equivalence}, we know that $-d_{\nu}(\cdot,\cdot)$ is a symmetric positive definite kernel induced by $d$. Therefore, we obtain the product kernel   $\mathcal{K}(\cdot,\cdot)$ \citep[ Lemma ~4.6 ]{steinwart2008support}. 
By Cauchy-Schwarz inequality and triangle inequality,
$$
 \mathbb{E} [\mathcal{K}^2(Z,Z')]\leq  C \mathbb{E}[d_{\nu}^4(X,X')] \leq C \mathbb{E}[d(\omega,X)+d(\omega,X')]^4<\infty.
$$
Then, by Proposition 1 and Theorem 2 in \cite{sun2005mercer}, the result follows.
\end{proof}

\begin{lem}\label{lem_mds}
Define $\mathcal{F}_{t}=\sigma(X_t,X_{t-1},\cdots)$, and $Z_{t}^{(k)}=(X_t,X_{t-k})$. Under $H_0$, for any fixed  $K$ and $M$, we have (i). $\left(\{e_{\ell}(Z_t^{(k)})\}_{\ell=1,\cdots,M; k=1,\cdots,K },\mathcal{F}_{t-1}\right)$ forms a 
sequence of martingale difference vectors; (ii). $\mathbb{E}[e_{\ell_1}(Z_t^{(k_1)})e_{\ell_2}(Z_t^{(k_2)})]=\mathbf{1}(\ell_1=\ell_2,k_1=k_2)$; (iii).  for $i_1<j_1, i_2<j_2$, $\mathbb{E}\left\{ e_{\ell_1}(Z_{i_1}^{(k_1)})e_{\ell_1}(Z_{j_1}^{(k_1)})e_{\ell_2}(Z_{i_2}^{(k_2)})e_{\ell_2}(Z_{j_2}^{(k_2)})\right\}=0$ except for the case when the index subset $\{i_1-k_1,i_1,j_1-k_1,j_1\}$ is identical to  $\{i_2-k_2,i_2,j_2-k_2,j_2\}$.
\end{lem}

\begin{proof}

(i). 
The M.D.S. claim follows by verifying for each $k$ and $\ell$. We first show that 
$$
\mathbb{E}[\mathcal{K}(z,Z_t^{(k)})|\mathcal{F}_{t-1}]=0,
$$
where $\mathcal{K}(z,z')$ is defined in \eqref{K_kernel} of \Cref{lem_expand}.

By direct calculation,
\begin{flalign*}
\mathbb{E}[\mathcal{K}(z,Z_t^{(k)})|\mathcal{F}_{t-1}]=& \mathbb{E}[d_{\nu}(x,X_t)d_{\nu}(y,X_{t-k})|\mathcal{F}_{t-1}]\\=&d_{\nu}(y,X_{t-k})\mathbb{E}[d_{\nu}(x,X_t)|\mathcal{F}_{t-1}]=d_{\nu}(y,X_{t-k})\mathbb{E}[d_{\nu}(x,X_t)]=0,
\end{flalign*}
where the second last equality holds by the independence between $X_t$ and $\mathcal{F}_{t-1}$, and the last by the fact that $\mathbb{E}[d_{\nu}(x,X_t)]=0$. 

By  \eqref{eq_eigen}, we have
$$
\lambda_{\ell}e_{\ell}(Z_t^{(k)})=\int \mathcal{K}(z,Z_t^{(k)}) e_\ell(z)\mathbb{P}_Z(\mathrm{d}z),
$$
taking conditional expectation w.r.t. $\mathcal{F}_{t-1}$ on both sides of the above equation we can obtain that 
\begin{equation}\label{eq_mds}
    \mathbb{E}[\lambda_{\ell}e_{\ell}(Z_t^{(k)})|\mathcal{F}_{t-1}]=\int \mathbb{E}[\mathcal{K}(z,Z_t^{(k)})|\mathcal{F}_{t-1}] e_\ell(z)\mathbb{P}_Z(\mathrm{d}z)=0.
\end{equation}
Note $\lambda_{\ell}$ is a constant, the result follows.

(ii). Under $H_0$, for any $k_1\neq k_2$, we have that 
\begin{flalign*}
\mathbb{E}[e_{\ell_1}(Z_t^{(k_1)})e_{\ell_2}(Z_t^{(k_2)})]=&\mathbb{E}[e_{\ell_1}(X_t,X_{t-k_1})e_{\ell_2}(X_t,X_{t-k_2})]\\=&\mathbb{E}[e_{\ell_1}(X_t,X_{t-k_1})e_{\ell_2}(X_t,X_{t+1})]\\
=&\mathbb{E}\{\mathbb{E}[e_{\ell_1}(X_t,X_{t-k_1})e_{\ell_2}(X_t,X_{t+1})|\mathcal{F}_{t}]\}
\\=&\mathbb{E}\{e_{\ell_1}(X_t,X_{t-k_1})\mathbb{E}[e_{\ell_2}(X_t,X_{t+1})|\mathcal{F}_{t}]\}=0
\end{flalign*}
where the second inequality holds by the fact that $\{X_t\}$ is i.i.d., and the third by law of iterated expectation and the last by \eqref{eq_mds}. For $k_1=k_2$ and $\ell_1\neq \ell_2$, we know that $\mathbb{E}[e_{\ell_1}(Z)e_{\ell_2}(Z)]=0$ by the orthogonality of eigenfunctions.

(iii). Note that 
\begin{flalign*}
&\mathbb{E}\left\{ e_{\ell_1}(Z_{i_1}^{(k_1)})e_{\ell_1}(Z_{j_1}^{(k_1)})e_{\ell_2}(Z_{i_2}^{(k_2)})e_{\ell_2}(Z_{j_2}^{(k_2)})\right\}
\\=&\mathbb{E}\left\{ e_{\ell_1}(X_{i_1},X_{i_1-k_1})e_{\ell_1}(X_{j_1},X_{j_1-k_1})e_{\ell_2}(X_{i_2},X_{i_2-k_2})e_{\ell_2}(X_{j_2},X_{j_2-k_2})\right\}.
\end{flalign*}
By similar arguments as (ii), we know that  if one element in $\{i_1-k_1,i_1,j_1-k_1,j_1, i_2-k_2,i_2,j_2-k_2,j_2\}$ is not paired, using the M.D.S. property we can claim the expectation is zero. In addition, when $i_1<j_1$ and $i_2<j_2$, we know that ${i_1-k_1}<{i_1},j_1-k_1<{j_1}$, and ${i_2-k_2}<{i_2},j_2-k_2<{j_2}$. This suggests we must have $i_1-k_1=i_2-k_2$, and $j_1=j_2$. When $k_1=k_2$, for the nonnull expectation, we must have $i_1=i_2$; and when $k_1\neq k_2$, we must have $i_2=j_1-k_1$ and $i_1=j_2-k_2$. In either case, the index subset $\{i_1-k_1,i_1,j_1-k_1,j_1\}$ is identical to  $\{i_2-k_2,i_2,j_2-k_2,j_2\}$.  

\end{proof}

\begin{lem}
\label{lem_exp}
Under $H_0$, $\mathbb{E}V_n(k)=O((n-k)^{-2})$.
\end{lem}
\begin{proof}
In view of \Cref{lem_expand} and \Cref{lem_mds} (ii), we first claim that  except for the case that $(i_1,i_2)=(j_1,j_2)$ or $(i_1,i_2)=(j_2,j_1)$,
\begin{equation}\label{prod_dnu}
    \mathbb{E}[d_{\nu}(X_{i_1},X_{i_2})d_{\nu}(X_{j_1},X_{j_2})]=0.
\end{equation}

By Hoeffding decomposition in \eqref{hoeffding}, under $H_0$, we have $V(k)=0$, $h_1=0$ and $\mathbb{E}h_2=0$,  $h^{(1)}=0$ and $\mathbb{E}h^{(2)}=0$. Therefore it suffices to consider $\mathbb{E}[h_c(Z_{i_1}^{(k)},\cdots,Z_{i_c}^{(k)})]$ for  $c=3,4,$ and  $k+1\leq i_1<\cdots,i_c\leq n$.  

For $c=3$, for $s<u<v$, we have
\begin{flalign}\label{h3}
\begin{split}
h_3(z_s,z_u,z_v)=&\frac{1}{12}\Big\{2\sum_{i_1<i_2}^{(s,u,v)}d_{\nu}(x_{i_1},x_{i_2})d_{\nu}(y_{i_1},y_{i_2})+2 \sum_{i}^{(s,u,v)}\mathbb{E}[d_{\nu}(x_i,X)d_{\nu}(y_i,Y)]\\&-\sum_{(i_1,i_2,i_3)}^{(s,u,v)}d_{\nu}(x_{i_1},x_{i_2})d_{\nu}(y_{i_1},y_{i_3})-\sum_{ i_1\neq i_2}^{(s,u,v)}\mathbb{E}[d_{\nu}(x_{i_1},X)d_{\nu}(y_{i_2},Y)]
\Big\}\\
=&h_{31}+h_{32}+h_{33}+h_{34}.
\end{split}
\end{flalign}
Clearly, by \eqref{prod_dnu}, $\mathbb{E}h_{31}=\mathbb{E}h_{32}=\mathbb{E}h_{34}=0$.  Recall for $z_{\ell}=Z_{\ell}^{(k)},$ $\ell\in(s,u,v)$, we have 
$$
h_{33}=-\frac{1}{12} \sum_{(i_1,i_2,i_3)}^{(s,u,v)}d_{\nu}(X_{i_1},X_{i_2})d_{\nu}(X_{i_1-k},X_{i_3-k}).
$$
where the nonnull expecation only appears in terms satisfying $(i_3-k,i_1-k)=(i_1,i_2)$, suggesting $s=u-k=v-2k$. Note there are at most $O(n-k)$ such $(s,u,v)$-tuples  that can satisfy the constraint. This implies that in the summands of U-statistics $U_{n,k}(h^{(3)})$, there are only $O(n-k)$ terms that have nonnull expectation, hence $\mathbb{E}U_{n,k}(h^{(3)})=O((n-k)^{-2})$.  Similar arguments also applies to $h_{4}$.

\end{proof}

\begin{lem}\label{lem_h0residual}
Under $H_0$, suppose \Cref{ass_moment} holds. Then, for $U_{n,k}(h^{(c)}), c=3,4,$ defined in \eqref{hoeffding}, with any fixed $k$, 
$$
\mathrm{Var}\left(U_{n,k}(h^{(c)})\right)\leq C(n-k)^{-3}, 
$$
for some uniform constant $C>0$ independent of $n$ and $k$.
\end{lem}

\begin{proof}
Note that $U_{n,k}(h^{(c)})$ is a U-statistic of order $c$ based on samples $Z_i^{(k)}=(X_i,X_{i-k})$, $k+1\leq i\leq n$. This implies that $\{Z_i^{(k)}\}$ is {\it $k$-dependent}. Following the treatment in \cite{janson2021asymptotic}, we define 
\begin{equation*}
    U_{n,k}(h^{(c)};>k)={n-k\choose c}^{-1}\sum_{\substack{k+1 \leqslant i_1<\cdots<i_{c}<n \\ i_{j+1}-i_j>k}} h^{(c)}\left(Z^{(k)}_{i_{1}}, \ldots, Z^{(k)}_{i_{c}}\right)
\end{equation*}
such that in above equation, $Z_{i_1}^{(k)},\cdots,Z_{i_c}^{(k)}$  are independent of each other. This implies that  $U_{n,k}(h^{(c)};>k)$ are formed based on independent copies of $Z_i^{(k)}$. By standard arguments in Hoeffding decomposion, see e.g., Theorem 3 in Chapter 1.6 of \cite{lee1990u}, we know that $\mathrm{Var}\left(U_{n,k}(h^{(c)};>k)\right)\leq C(n-k)^{-c}$.
Furthermore, by Lemma 4.4 in \cite{janson2021asymptotic}, when $\mathbb{E}h^2(\cdot)<\infty$ (which is ensure by Assumption \ref{ass_moment}),  we have that
$\mathrm{Var}\left(U_{n,k}(h^{(c)})-U_{n,k}(h^{(c)};>k)\right)\leq C(n-k)^{-3}$.  Hence, by Cauchy-Schwarz inequality,   we know that $$
\mathrm{Var}\left(U_{n,k}(h^{(c)})\right)\leq 2\max\left\{\mathrm{Var}\left(U_{n,k}(h^{(c)})-U_{n,k}(h^{(c)};>k)\right),\mathrm{Var}\left(U_{n,k}(h^{(c)};>k))\right)\right\}\leq 2C(n-k)^{-3}.
$$

\end{proof}


\section{Technical Proofs}

\noindent\textbf{Proof of Theorem \ref{thm_fix}}\label{Sec:proof}

(i). Leading term.

Under $H_0$, we have $V_X(k)=0$, $h_1(z)=0$ and $h^{(2)}(z,z')=h_2(z,z')=d_{\nu}(x,x')d_{\nu}(y,y')/6=\mathcal{K}(z,z')/6$,  with the kernel $\mathcal{K}(z,z')$ defined in \eqref{K_kernel}. Hence, by Hoeffding decomposition in Lemma \ref{lem_hoeff},
\begin{flalign*}
    \notag V_n(k) = & {n-k\choose 2}^{-1}\sum_{i=k+1}^{n}\sum_{j=i+1}^n \mathcal{K}(Z_i^{(k)},Z_j^{(k)})+\sum_{c=3}^4{4\choose c}U_{n,k}(h^{(c)}).
\end{flalign*}
By \Cref{lem_exp} and \Cref{lem_h0residual}, we know that $\mathbb{E}[U_{n,k}(h^{(c)})]^2=O((n-k)^{-3})$ for $c=3,4$, hence by Markov inequality, 
\begin{flalign}\label{Vtilde}
 V_n(k) = & {n-k\choose 2}^{-1}\sum_{i=k+1}^{n}\sum_{j=i+1}^n \mathcal{K}(Z_i^{(k)},Z_j^{(k)})+o_p(n^{-1}):= \tilde{V}_n(k)+o_p(n^{-1}).
\end{flalign}
Since $K$ is fixed, in the following proof, with minor abuse of notation, we can approximately consider 
$$
\tilde{V}_n(k)= {n-K\choose 2}^{-1}\sum_{i=K+1}^{n}\sum_{j=i+1}^n \mathcal{K}(Z_i^{(k)},Z_j^{(k)}).
$$

(ii). Approximation for $\mathcal{K}(z,z')$.

By Lemma \ref{lem_expand}, let
$$
\mathcal{K}^{(M)}(z,z')=\sum_{\ell=1}^{M}\lambda_{\ell} e_{\ell}(z)e_{\ell}(z')\to \mathcal{K}(z,z'),\quad \mbox{ as  } M\to\infty.
$$
Let 
\begin{flalign}
     \label{VM}\tilde{V}_n^{(M)}(k) =&{n-K\choose 2}^{-1} \sum_{i=K+1}^n\sum_{j=i+1}^n\mathcal{K}^{(M)}(Z_i^{(k)},Z_j^{(k)})\\\notag=&{n-K\choose 2}^{-1} \sum_{i=K+1}^n\sum_{j=i+1}^n\sum_{\ell=1}^{M}\lambda_{\ell} e_{\ell}(Z_i^{(k)})e_{\ell}(Z_j^{(k)})
\end{flalign}
Then, we have that, 
\begin{flalign}\label{Vk1}
\begin{split}
& n^2\mathbb{E}[\tilde{V}_{n}(k)-\tilde{V}_{n}^{(M)}(k)]^2 \\= & n^2 {n-K\choose 2}^{-2}\mathbb{E}\left[ \sum_{i=K+1}^n\sum_{j=i+1}^n\sum_{\ell=M+1}^{\infty}\lambda_{\ell} e_{\ell}(Z_i^{(k)})e_{\ell}(Z_j^{(k)})\right]^2 \\
 =&n^2{n-K\choose 2}^{-2}\mathbb{E}\left[ \sum_{i=K+1}^n\sum_{j=i+1}^n\sum_{\ell_1,\ell_2=M+1}^{\infty}\lambda_{\ell_1}\lambda_{\ell_2}  e^2_{\ell_1}(Z_i^{(k)})e^2_{\ell_2}(Z_j^{(k)})\right]
 \\\leq &C \sum_{\ell_1,\ell_2=M+1}^{\infty}\lambda_{\ell_1}\lambda_{\ell_2}\to 0,\quad \mbox{as } M\to\infty.
\end{split}
\end{flalign}
where the second equality holds by  Lemma \ref{lem_mds} (iii), and the inequality holds by Cauchy-Schwarz inequality $\mathbb{E}e^2_{\ell_1}(Z_i^{(k)})e^2_{\ell_2}(Z_j^{(k)})\leq [\mathbb{E}e^4_{\ell_1}(Z)]^{1/2}[\mathbb{E}e^4_{\ell_2}(Z)]^{1/2}\leq \sup_{\ell}\mathbb{E}[e_{\ell}^4(Z)]<\infty$, which is ensured by $\mathbb{E}(\mathcal{K}^2(Z,Z))<\infty$. The final convergence is ensured by $\sum_{\ell=1}^{\infty}\lambda_{\ell}=\mathbb{E}\mathcal{K}(Z,Z)<\infty$.

(iii). Joint convergence of $n\left\{\tilde{V}_{n}^{(M)}(k)\right\}_{k=1}^K$.

Note by \eqref{VM}, we have
\begin{flalign*}
n\tilde{V}_{n}^{(M)}(k)
=& \frac{n}{n-K-1}\sum_{\ell=1}^{M} \lambda_{\ell} \left\{\left[\frac{1}{\sqrt{n-K}}\sum_{i=K+1}^n e_{\ell}(Z_{i}^{(k)})\right]^2 -\frac{1}{{n-K}}\sum_{i=K+1}^n e_{\ell}^2(Z_{i}^{(k)})\right\}
\end{flalign*}

In view of \Cref{lem_mds} (i), we know that  $\left(\{e_{\ell}(Z_t^{(k)}\}_{\ell=1,\cdots,M;k=1,\cdots,K}),\mathcal{F}_{t-1}\right)$ forms a sequence of martingale difference for any fixed $M$ and $K$.     Hence $\frac{1}{{n-K}}\sum_{i=K+1}^n e_{\ell}^2(Z_{i}^{(k)})\to_p \mathbb{E}e_{\ell}^2(Z)=1$ by the weak law of large numbers, and
$$
\left\{\frac{1}{\sqrt{n-K}}\sum_{i=K+1}^n e_{\ell}(Z_{i}^{(k)})\right\}_{\ell=1,\cdots, M; k=1\cdots,K}\to_d \{G^{(k)}_{\ell}\}_{\ell=1,\cdots, M; k=1\cdots,K} 
$$
by CLT for martingale difference sequences and Cram\'er Wold device, where $\{G^{(k)}_{\ell}\}_{\ell=1,\cdots, M; k=1\cdots,K}$ is a sequence of i.i.d. standard normal distribution in view of \Cref{lem_mds} (ii).



Then continuous mapping theorem implies that 
$$
n\left(\tilde{V}_{n}^{(M)}(k)\right)_{k=1}^K\to_d \sum_{\ell=1}^M \lambda_{\ell} \{[G_{\ell}^{(k)}]^2-1\}.
$$

The final result then follows by letting $M\to\infty$ in view of \eqref{Vk1}.

\vspace{5mm}
\noindent\textbf{Proof of Theorem \ref{thm_spec}}

Denote 
\begin{equation}\label{Sn_decomp}
   S_n(\zeta)= \sum_{k=1}^{K}(n-k)V_n(k)\Psi_k(\zeta)+\sum_{k=K+1}^{n-4}(n-k)V_n(k)\Psi_k(\zeta):= S_n^K(\zeta)+R_n^K(\zeta),
\end{equation}
and for $\{\xi_k\}$ defined in Theorem \ref{thm_fix},
$$
S^K(\zeta)= \sum_{k=1}^{K}\xi_k\Psi_k(\zeta).
$$

By Proposition 6.3.9 of \cite{brockwell1991time}, to show the weak convergence of $S_n(\zeta)$, it suffices to show that 
(i) for each $K$, $S_n^K(\zeta)\Rightarrow S^K(\zeta)$ as $n\to\infty$; (ii) $S^K(\zeta)\Rightarrow S(\zeta)$ as $K\to\infty$; (iii) for any $\epsilon>0$, $\lim_{K\to\infty}\lim_{n\to\infty }\mathbb{P}(\|R_{n}^K\|>\epsilon)=0$.

(i). Note that  (i) follows by showing (a) the finite dimensional convergence of  $\{S_n^K(\zeta_i)\}_{i=1}^M$ for arbitrary finite $M$, $(\zeta_1,\zeta_2,\cdots,\zeta_M)$ with $\zeta_i\in[0,\pi]$; (b)  the tightness of  $S_n^K(\zeta)$.

First, as $K$ and $M$ are finite, the proof of (a) follows directly from Theorem \ref{thm_fix} and continuous mapping theorem. 

Second, for (b), we note that $ S_{n}^K(\zeta)= \sum_{k=1}^{K}(n-k){V}_n(k)\Psi_k(\zeta)$, and when $K$ is fixed, we only need to show the tightness of $(n-k){V}_n(k)\Psi_k(\zeta)$ for each $k$. 
Fix a complete orthonormal basis in $\mathbb{H}$, denoted by $e_1,e_2,\cdots$. By Lemma 7.1 in \cite{panaretos2013fourier}, it suffices to show that 
$$
\limsup_n \sum_{j=1}^{\infty}\mathbb{E}|\langle (n-k){V}_n(k)\Psi_k,e_j\rangle|^2 <\infty.
$$
In fact, for any $h\in \mathbb{H}$, denote $W_h(k)=\int_{[0,\pi]}h(\zeta)\Psi_k(\zeta)\mathrm{d}\zeta$.  Then 
\begin{flalign*}
    \mathbb{E}|\langle (n-k){V}_n(k),h\rangle|^2
=\mathbb{E}[(n-k){V}_n(k)]^2W_h^2(k)\leq CW_h^2(k),
\end{flalign*}
where we note that under Assumption \ref{ass_moment},  for any $n\geq k+4$, we have $\mathbb{E}[(n-k){V}_n(k)]^2\leq C<\infty.$
Hence, 
$$
\sum_{j=1}^{\infty}\mathbb{E}|\langle (n-k){V}_n(k)\Psi_k,e_j\rangle|^2 \leq C \sum_{j=1}^{\infty}W_{e_j}^2(k)=C\|\Psi_k\|^2<\infty.
$$




Hence, (b) holds, and this completes the proof for (i).

(ii) Holds trivially.

(iii) Recall $$
    R_{n}^K(\zeta)= \sum_{k=K+1}^{n-4}(n-k){V}_n(k)\Psi_k(\zeta).
$$
By Chebyshev inequality,  we have 
\begin{flalign*}
    &\mathbb{P}(\|R_{n,1}^K\|>\epsilon)\leq \epsilon^{-2}\mathbb{E}\|R_{n}^K\|^2= \epsilon^{-2} 
    \mathbb{E}\left[\sum_{k=K+1}^{n-4}(n-k)^2[{V}_n(k)]^2 \|\Psi_k\|^2\right]\leq 4C\epsilon^{-2}K^{-1},
\end{flalign*}
where the equality holds by noting that $\langle \Psi_{k_1},\Psi_{k_2}\rangle=0$ if $k_1\neq k_2$ and the last inequality holds using $\|\Psi_k\|^2\leq k^{-2}$.

Hence, (iii) follows by letting $K\to\infty$. This completes the proof.

\qed 
\vspace{5mm}

\noindent\textbf{Proof of  \Cref{thm_consis}}

 It suffices to show that $V_n(k)\to_p V(k)$. 
Recall the alternative representation in \eqref{kernel}, in view of   Theorem 1 (iii) in \cite{arcones1998law}, it suffices to show that for any $i,j,q,r$, for some $0<\delta\leq 1$, 
$$
\sup_{1\leq i,j,q,r<\infty} \mathbb{E}[|h(Z_i^{(k)},Z_j^{(k)},Z_q^{(k)},Z_r^{(k)})|(\log^+|h(Z_i^{(k)},Z_j^{(k)},Z_q^{(k)},Z_r^{(k)})|)^{(1+\delta)}]<\infty.
$$
For any $x\in\mathbb{R}$, by the elementary inequality that $\log(1+x)\leq x$ for $x\geq0$, we have 
$$
(\log^+|x|)^{1+\delta}=\left[\frac{1+\delta}{\delta}\log^+|x|^{\delta/(1+\delta)}\right]^{(1+\delta)}\leq \left[\frac{1+\delta}{\delta}\log(1+|x|^{\delta/(1+\delta)})\right]^{(1+\delta)}\leq  \left[\frac{1+\delta}{\delta}|x|^{\delta/(1+\delta)}\right]^{(1+\delta)}.
$$
Therefore, it suffices to show that for some $\delta>0$,
$$
\sup_{1\leq i,j,q,r<\infty} \mathbb{E}[|h(Z_i^{(k)},Z_j^{(k)},Z_q^{(k)},Z_r^{(k)})|^{1+\delta}]<\infty.
$$
In fact,
\begin{flalign}
\notag\sup_{i,j,q,r}E\{|h(Z_i^{(k)},Z_j^{(k)},Z_q^{(k)},Z_r^{(k)})|^{1+\delta}\}\leq &C \sup_{(i_1,i_2,i_3,i_4)}\left(E\left|a_{i_1 i_2} b_{i_3i_4}\right|^{1+\delta}+E\left|a_{i_1i_2} b_{i_1i_2}\right|^{1+\delta}+2E\left|a_{i_1i_2} b_{i_1i_3}\right|^{1+\delta}\right)\\\notag
\leq & C \sup_{i_1,i_2}Ed^{2+2\delta}(X_{i_1},X_{i_2})
\\\label{bound_h}\leq &  C \sup_{i_1,i_2} \mathbb{E}[|d(X_{i_1},\omega)+d(X_{i_2},\omega)|^{2+2\delta}]<\infty.
\end{flalign}
where the first inequality holds  by Minkowski inequality, the second by H\"older inequality, and the third by triangle inequality.

\qed 
\vspace{5mm}

\noindent\textbf{Proof of \Cref{thm_boot_fix}}

(i).  Leading term by  $ \tilde{V}_n^{*}(k) ={n-k\choose 2}^{-1} \sum_{i=k+1}^n\sum_{j=i+1}^n\mathcal{K}(Z_i^{(k)},Z_j^{(k)})w_i(k)w_j(k)$.

Recall  $\mathcal{K}(Z_i^{(k)},Z_j^{(k)})=d_{\nu}(X_i,X_{j})d_{\nu}(X_{i-k},X_{j-k}),$ with $\mathcal{K}(z,z')$ defined in \eqref{K_kernel}. It is clear that
$$\mathbb{E}^*\left[ \frac{1}{n-k-3}\sum_{k+1\leq i\neq j\leq n}\left\{\tilde{a}_{ij}\tilde{b}_{ij}-\mathcal{K}(Z_i^{(k)},Z_j^{(k)})\right\}w_i(k)w_j(k)\right]=0,$$ so we first want to show that 
\begin{flalign}\label{boot_leading}
\begin{split}
       &\mathrm{Var}^*\left[ \frac{1}{n-k-3}\sum_{k+1\leq i\neq j\leq n}\left\{\tilde{a}_{ij}\tilde{b}_{ij}-\mathcal{K}(Z_i^{(k)},Z_j^{(k)})\right\}w_i(k)w_j(k)\right]\\=&\frac{2}{(n-k-3)^2}\sum_{k+1\leq i\neq j\leq n}\left\{\tilde{a}_{ij}\tilde{b}_{ij}-\mathcal{K}(Z_i^{(k)},Z_j^{(k)})\right\}^2\to_p 0, 
\end{split}
\end{flalign}
where the equality holds using the joint independence of $\{w_i(k)\}$.

By Cauchy-Schwarz and H\"older  inequality, we have that \begin{flalign*}
   &\sum_{k+1\leq i\neq j\leq n}\left\{\tilde{a}_{ij}\tilde{b}_{ij}-\mathcal{K}(Z_i^{(k)},Z_j^{(k)})\right\}^2\\
   \leq &   2\sum_{k+1\leq i\neq j\leq n}[\tilde{a}_{ij}-d_{\nu}(X_i,X_{j})]^2\tilde{b}^2_{ij} +2\sum_{k+1\leq i\neq j\leq n}[\tilde{b}_{ij}-d_{\nu}(X_{i-k},X_{j-k})]^2d^2_{\nu}(X_i,X_{j})\\
   \leq &   4\sum_{k+1\leq i\neq j\leq n}[\tilde{a}_{ij}-d_{\nu}(X_i,X_{j})]^2[\tilde{b}_{ij}-d_{\nu}(X_{i-k},X_{j-k})]^2 \\
   &+ 4\sum_{k+1\leq i\neq j\leq n}[\tilde{a}_{ij}-d_{\nu}(X_i,X_{j})]^2d^2_{\nu}(X_{i-k},X_{j-k}) \\&+2\sum_{k+1\leq i\neq j\leq n}[\tilde{b}_{ij}-d_{\nu}(X_{i-k},X_{j-k})]^2d^2_{\nu}(X_i,X_{j})\\
   \leq & 4\left(\sum_{k+1\leq i\neq j\leq n}[\tilde{a}_{ij}-d_{\nu}(X_i,X_{j})]^4\right)^{1/2}\left(\sum_{k+1\leq i\neq j\leq n}[\tilde{b}_{ij}-d_{\nu}(X_{i-k},X_{j-k})]^4\right)^{1/2} 
   \\&+4\left(\sum_{k+1\leq i\neq j\leq n}[\tilde{a}_{ij}-d_{\nu}(X_i,X_{j})]^4\right)^{1/2}\left(\sum_{k+1\leq i\neq j\leq n}d^4_{\nu}(X_{i-k},X_{j-k})\right)^{1/2} 
   \\&+2\left(\sum_{k+1\leq i\neq j\leq n}[\tilde{b}_{ij}-d_{\nu}(X_{i-k},X_{j-k})]^4\right)^{1/2} \left(\sum_{k+1\leq i\neq j\leq n}d^4_{\nu}(X_i,X_{j})\right)^{1/2}.
\end{flalign*}
Under Assumption \ref{ass_moment}, it is clear that $\frac{1}{(n-k-3)^2}\sum_{k+1\leq i\neq j\leq n}d^4_{\nu}(X_i,X_{j})$ and $\frac{1}{(n-k-3)^2}\sum_{k+1\leq i\neq j\leq n}d^4_{\nu}(X_{i-k},X_{j-k})$ are bounded in probability. This implies we only need to show 
\begin{equation}\label{eq_sufU}
  \frac{1}{(n-k-3)^2}\sum_{k+1\leq i\neq j\leq n}[\tilde{a}_{ij}-d_{\nu}(X_i,X_{j})]^4\to_p 0
\end{equation}
as $$
 \quad \frac{1}{(n-k-3)^2}\sum_{k+1\leq i\neq j\leq n}[\tilde{b}_{ij}-d_{\nu}(X_{i-k},X_{j-k})]^4\to_p 0$$ is similar.
By Minkowski inequality, we have 
\begin{flalign*}
    &\mathbb{E}[\tilde{a}_{ij}-d_{\nu}(X_i,X_{j})]^4 \\\leq& C\left\{\mathbb{E}\left[\frac{\sum_{t=k+1}^n d(X_t,X_i)}{n-k-2}-D^{(1)}(X_i)\right]^4+\mathbb{E}\left[\frac{\sum_{k+1\leq t,t'\leq n} d(X_t,X_{t'})}{(n-k-1)^2}-D\right]^4 \right\}.
\end{flalign*}

For the first term,  under $H_0$, we note that  
\begin{flalign*}
  \mathbb{E}\prod_{s=1}^4 [d(X_{i_s},X_i)-D^{(1)}(X_i)]= \mathbb{E}\left\{\mathbb{E}\left(\prod_{s=1}^4 [d(X_{i_s},X_i)-D^{(1)}(X_i)]\big|X_i\right)\right\}=0
\end{flalign*}
for distinct 4-tuples $(i_1,i_2,i_3,i_4)$. This implies that the first term is at most of order $O(n^{-1})$ in view of Assumption \ref{ass_moment}.  Similarly, the second term is at most of order $O(n^{-1})$. Hence $\mathbb{E}[\tilde{a}_{ij}-d_{\nu}(X_i,X_{j})]^4 \to 0$, and by Markov  inequality, \eqref{eq_sufU} holds. Therefore \eqref{boot_leading} holds.

(ii). Approximation.

Denote \begin{equation}\label{tildeV_boot}
    \tilde{V}_n^{X,(M)*}(k) ={n-k\choose 2}^{-1} \sum_{i=k+1}^n\sum_{j=i+1}^n\mathcal{K}^{(M)}(Z_i^{(k)},Z_j^{(k)})w_i(k)w_j(k).
\end{equation}
We want to show that as $M\to\infty$,
\begin{equation}\label{boot_approximation}
    \mathbb{E}^*\left[n|\tilde{V}_n^{X,(M)*}(k)-\tilde{V}_n^{*}(k)|\right]^2\to_p 0 
\end{equation}
Recall $\mathcal{K}^{(M)}(z,z')=\sum_{\ell=1}^M \lambda_{\ell}e_{\ell}(z)e_{\ell}(z')$, we have
\begin{flalign*}
    &\mathbb{E}^*\left[n{n-k \choose 2}^{-1}\sum_{k+1\leq i<j\leq n}\left\{\mathcal{K}^{(M)}(Z_i^{(k)},Z_j^{(k)})-\mathcal{K}(Z_i^{(k)},Z_j^{(k)})\right\}w_i(k)w_j(k)\right]^2  \\
 =&n^2{n-k\choose 2}^{-2}\sum_{k+1\leq i<j\leq n}\left[ \sum_{\ell=M+1}^\infty \lambda_{\ell}e_{\ell}(Z_i^{(k)})e_{\ell}(Z_j^{(k)})\right]^2  
 \\=& {n-k\choose 2}^{-1}\sum_{\substack{k+1\leq i<j\leq n\\ j>i+k}}\left[ \sum_{\ell=M+1}^\infty \lambda_{\ell}e_{\ell}(Z_i^{(k)})e_{\ell}(Z_j^{(k)})\right]^2 \\& +{n-k\choose 2}^{-1}\sum_{\substack{k+1\leq i<j\leq n\\ j\leq i+k}}\left[ \sum_{\ell=M+1}^\infty \lambda_{\ell}e_{\ell}(Z_i^{(k)})e_{\ell}(Z_j^{(k)})\right]^2+o_p(1)\\
 \to_p&\mathbb{E}\left[ \sum_{\ell=M+1}^\infty \lambda_{\ell}e_{\ell}(Z)e_{\ell}(Z')\right]^2=\sum_{\ell=M+1}^{\infty}\lambda_{\ell}^2
\end{flalign*}
where we note that $Z_i^{(k)}$ and $Z_j^{(k)}$  are independent for $|i-j|>k$, and that there are at most $kn$ terms in the summation of $\sum_{\substack{k+1\leq i<j\leq n\\ j\leq i+k}}$.

Hence, \eqref{boot_approximation} follows by that $\sum_{\ell=M+1}^{\infty}\lambda_{\ell}^2\to0$ as $M\to\infty$.

(iii). Joint Convergence of $\{(n-k)\tilde{V}_n^{X,(M)*}(k)\}_{k=1}^K\to_{d^*} \{\sum_{\ell=1}^{M}\lambda_{\ell}[G_{\ell}^{(k)}]^2-1\}_{k=1}^K$ in probability.

By continuous mapping theorem, it suffices to show that, in probability,
\begin{flalign*}
   & \frac{1}{n-k}\sum_{i=k+1}^n w_i^2(k) e_{\ell}^2(Z_i^{(k)})\to_{p^*}1,
\\
&\left\{\frac{1}{\sqrt{n-k}}\sum_{i=k+1}^n w_i(k) e_{\ell}(Z_i^{(k)})\right\}_{\ell=1,\cdots,M; k=1,\cdots,K}\to_{d^*} \left\{G_{\ell}^{(k)}\right\}_{\ell=1,\cdots,M; k=1,\cdots,K}.
\end{flalign*}

Note that  $\mathbb{E}^*[\frac{1}{n-k}\sum_{i=k+1}^n w_i^2(k) e_{\ell}^2(Z_i^{(k)})]=\frac{1}{n-k}\sum_{i=k+1}^n e_{\ell}^2(Z_i^{(k)})\to_p 1$, and 
$
\mathrm{Var}^*[\frac{1}{n-k}\sum_{i=k+1}^n w_i^2 (k)e_{\ell}^2(Z_i^{(k)})]=\frac{1}{(n-k)^2}\sum_{i=k+1}^n e_{\ell}^4(Z_i^{(k)}) \mathrm{Var}(w_i^2(k))\to_p 0
$
by weak law of large numbers for $k-$dependent sequences and Slutsky's theorem. Hence, by Chebyshev inequality, $$
\frac{1}{n-k}\sum_{i=k+1}^n w_i^2(k) e_{\ell}^2(Z_i^{(k)})\to_{p^*}1,\quad \mbox{in probability}.$$

Next, note that for any fixed $\ell$ and $k$,  $\sum_{i=k+1}^n [(n-k)^{-1/2}e_{\ell}(Z_i^{(k)})]^2\to_p 1$ and
$$\frac{\sum_{i=k+1}^n [(n-k)^{-1/2}e_{\ell}(Z_i^{(k)})]^4}{\{\sum_{i=k+1}^n [(n-k)^{-1/2}e_{\ell}(Z_i^{(k)})]^2\}^2}\to_p 0,$$
which implies that Lyapunov central limit theorem holds in probability.

Furthermore, for any fixed $\ell_1, \ell_2$ and $k_1\leq k_2$, using the independence between $\mathbf{w}(k_1)$ and $\mathbf{w}(k_2)$ for $k_1\neq k_2$, we have 
\begin{flalign*}
  &\mathrm{Cov}^*(\frac{1}{\sqrt{n-k_1}}\sum_{i=k+1}^n w_i(k_1)e_{\ell_1}(Z_i^{(k_1)}),\frac{1}{\sqrt{n-k_2}}\sum_{i=k+1}^n w_i(k_2) e_{\ell_2}(Z_i^{(k_2)}))\\=&\frac{1}{n-k_2}\sum_{i=k_2+1}^n e_{\ell_1}(Z_i^{(k_2)})e_{\ell_2}(Z_i^{(k_1)})\mathbf{1}(k_1=k_2)
  \to_p \mathbf{1}(\ell_1=\ell_2,k_1=k_2)
\end{flalign*}
in view of \Cref{lem_mds}(ii).

Therefore, by Cram\'er-Wold device, we obtain the joint convergence. 

Finally, in view of \eqref{boot_leading} and \eqref{boot_approximation}, the result follows.

\qed
\vspace{5mm}

\noindent\textbf{Proof of \Cref{thm_boot_test}}
 
The proof is similar to \Cref{thm_spec}. By continuous mapping theorem, it suffices to show that 
$$
\{S_n^*(\zeta)\}_{\zeta\in[0,\pi]}\Rightarrow^* \{S(\zeta)\}_{\zeta\in[0,\pi]},\quad \mbox{in probability},
$$
where $\Rightarrow^*$ represents the weak convergence in $L_2[0,\pi]$ under bootstrap asymptotics.

Denote 
\begin{equation}\label{Sn_decomp_boot}
   S_n^*(\zeta)= \sum_{k=1}^{K}(n-k)V_n^{*}(k)\Psi_k(\zeta)+\sum_{k=K+1}^{n-4}(n-k)V_n^{*}(k)\Psi_k(\zeta):= S_n^{K*}(\zeta)+R_n^{K*}(\zeta).
\end{equation}
To show the weak convergence of $S_n^*(\zeta)$ in probability, it suffices to show  in probability
(i). for each $K$, $S_n^{K*}(\zeta)\Rightarrow S^K(\zeta)$ as $n\to\infty$; (ii). $S^{K}(\zeta)\Rightarrow S(\zeta)$ as $K\to\infty$; (iii). for any $\epsilon>0$, $\lim_{K\to\infty}\lim_{n\to\infty }\mathbb{P}^*(\|R_{n}^{K*}\|>\epsilon)=0$.  

(i) We need to show that (a). $\{S_n^{K*}(\zeta_i)\}_{i=1}^M\to_d \{S^K(\zeta_i)\}_{i=1}^M$ in probability;  (b). $S_n^{K*}(\zeta)$ is asymptotically tight conditional on the sample. 



The proof of (a) follows directly from Theorem \ref{thm_boot_fix} and continuous mapping theorem. As for $(b)$, note that $K$ is fixed, we thus only need to show the tightness of $(n-k)V_n^{*}(k)\Psi_k(\cdot)$ for each $1\leq k\leq K$, which is easily ensured if $\mathbb{E}^*[(n-k)V_n^{*}(k)]^2$ is finite. In fact, 
\begin{equation}\label{V_square_boot}
    \mathbb{E}^*[(n-k)V_n^{*}(k)]^2= 2(n-k-3)^{-2}\sum_{k+1\leq i\neq j\leq n}[\tilde{a}_{ij}(k)\tilde{b}_{ij}(k)]^2,  
\end{equation}
and using Minkowski inequality,  \begin{flalign}\label{bound_a}
    \mathbb{E}[\tilde{a}_{ij}\tilde{b}_{ij}(k)]^2\leq \mathbb{E}[\tilde{a}^4_{ij}]^{1/2}\mathbb{E}[\tilde{b}^4_{ij}(k)]^{1/2}\leq CEd^4(\omega,X)<\infty. 
\end{flalign}
Therefore, $\mathbb{E}\mathbb{E}^*[(n-k)V_n^{*}(k)]^2<\infty$, which implies that in probability, $\mathbb{E}^*[(n-k)V_n^{*}(k)]^2$ is finite. And thus, (i) is proved.

(ii) is trivially satisfied.

(iii)   Note that under $\mathbb{P}^*$, $V_n^{*}(k_1)$ is independent of $V_n^{*}(k_2)$ for $k_1\neq k_2$. Therefore, by Chebyshev's inequality, 
$$
\mathbb{P}^*(\|R_{n}^{K*}\|>\epsilon)\leq \epsilon^{-2} \mathbb{E}^*\|R_{n}^{K*}\|^2= \epsilon^{-2} \sum_{k=K+1}^n[(n-k)V_n^{*}(k)]^2 \|\Psi_k\|^2.
$$
By similar arguments used in proving (i)(b) above, and that $\|\Psi_k\|^2\leq k^{-2}$, we have 
$\mathbb{P}^*(\|R_{n}^{K*}\|>\epsilon)=O_p(K^{-1})$. Therefore, letting $K\to\infty$ the results follows. 
\qed

\vspace{5mm}
\noindent\textbf{Proof of \Cref{thm_boot_consis}}
 By \Cref{thm_consis},  it suffices to show that $CvM^*_n=O^*_p(1)$  in probability, where $O_p^*(1)$ is analogous to $O_p(1)$ but for bootstrap sample asymptotics.

 Note that $CvM^*_n>0$, hence by Markov inequality, we only need to show $\mathbb{E}\mathbb{E}^*( \textsc{CvM}_n^*)=O(1)$. In fact, using the fact that  $\mathbb{E}^*V_n^{*}(k)=0$, and $\mathrm{Cov}^*(V_n^{*}(k),V_n^{*}(k'))=0$ for $k\neq k'$, we have 
\begin{flalign*}
     \mathbb{E}^*( \textsc{CvM}_n^*)=&\sum_{k=1}^{n-4}(n-k)^2\mathbb{E}^*({V}_n^{*}(k))^2\|\Psi_k\|^2
    \\=&2\sum_{k=1}^{n-4}(n-k-3)^{-2}\sum_{k+1\leq i\neq j\leq n}[\tilde{a}_{ij}(k)\tilde{b}_{ij}(k)]^2\|\Psi_k\|^2.
\end{flalign*}
By \eqref{bound_a}, we have
$
 \mathbb{E}\mathbb{E}^*( \textsc{CvM}_n^*)\leq
C \sum_{k=1}^{n-4} k^{-2}\leq C.
$
This implies that $ \mathbb{E}^*( \textsc{CvM}_n^*)$ is bounded in probability, the result follows.

\qed

\vspace{10pt}
\bibliographystyle{apalike}
\bibliography{arxiv_V2}

\begin{thebibliography}{}

\bibitem[A{\"\i}t-Sahalia and Yu, 2009]{ait2009high}
A{\"\i}t-Sahalia, Y. and Yu, J. (2009).
\newblock High frequency market microstructure noise estimates and liquidity
  measures.
\newblock {\em Annals of Applied Statistics}, 3(1):422--457.

\bibitem[Arcones, 1998]{arcones1998law}
Arcones, M. (1998).
\newblock The law of large numbers for {$ U $}-statistics under absolute
  regularity.
\newblock {\em Electronic Communications in Probability}, 3:13--19.

\bibitem[Aue et~al., 2017]{aue2017functional}
Aue, A., Horv{\'a}th, L., and F.~Pellatt, D. (2017).
\newblock Functional generalized autoregressive conditional heteroskedasticity.
\newblock {\em Journal of Time Series Analysis}, 38(1):3--21.

\bibitem[Board and Meyer-ter Vehn, 2021]{board2021learning}
Board, S. and Meyer-ter Vehn, M. (2021).
\newblock Learning dynamics in social networks.
\newblock {\em Econometrica}, 89(6):2601--2635.

\bibitem[Box and Pierce, 1970]{box1970distribution}
Box, G.~E. and Pierce, D.~A. (1970).
\newblock Distribution of residual autocorrelations in
  autoregressive-integrated moving average time series models.
\newblock {\em Journal of the American Statistical Association},
  65(332):1509--1526.

\bibitem[Brockwell et~al., 1991]{brockwell1991time}
Brockwell, P.~J., Davis, R.~A., and Fienberg, S.~E. (1991).
\newblock {\em Time Series: Theory and Methods: Theory and Methods}.
\newblock Springer Science \& Business Media.

\bibitem[Cerovecki et~al., 2019]{cerovecki2019functional}
Cerovecki, C., Francq, C., H{\"o}rmann, S., and Zako{\"\i}an, J.-M. (2019).
\newblock Functional {GARCH} models: The quasi-likelihood approach and its
  applications.
\newblock {\em Journal of Econometrics}, 209(2):353--375.

\bibitem[Davis et~al., 2018]{davis2018applications}
Davis, R., Matsui, M., Mikosch, T., and Wan, P. (2018).
\newblock Applications of distance covariance to time series.
\newblock {\em Bernoulli}, 24(4A):3087--3116.

\bibitem[Dehling and Mikosch, 1994]{dehling1994random}
Dehling, H. and Mikosch, T. (1994).
\newblock Random quadratic forms and the bootstrap for {U}-statistics.
\newblock {\em Journal of Multivariate Analysis}, 51(2):392--413.

\bibitem[Deo, 2000]{deo2000spectral}
Deo, R.~S. (2000).
\newblock Spectral tests of the martingale hypothesis under conditional
  heteroscedasticity.
\newblock {\em Journal of Econometrics}, 99(2):291--315.

\bibitem[Dubey and M{\"u}ller, 2019]{dubey2019frechet}
Dubey, P. and M{\"u}ller, H.-G. (2019).
\newblock Fr{\'e}chet analysis of variance for random objects.
\newblock {\em Biometrika}, 106(4):803--821.

\bibitem[Dubey and M\"uller, 2020]{dubeymuller2020}
Dubey, P. and M\"uller, H.-G. (2020).
\newblock Fr{\'e}chet change point detection.
\newblock {\em Annals of Statistics}, 48(6):3312--3335.

\bibitem[Escanciano and Lobato, 2009]{escanciano2009automatic}
Escanciano, J.~C. and Lobato, I.~N. (2009).
\newblock An automatic portmanteau test for serial correlation.
\newblock {\em Journal of Econometrics}, 151(2):140--149.

\bibitem[Escanciano and Velasco, 2006]{escanciano2006generalized}
Escanciano, J.~C. and Velasco, C. (2006).
\newblock Generalized spectral tests for the martingale difference hypothesis.
\newblock {\em Journal of Econometrics}, 134(1):151--185.

\bibitem[Fokianos and Pitsillou, 2017]{fokianos2017consistent}
Fokianos, K. and Pitsillou, M. (2017).
\newblock Consistent testing for pairwise dependence in time series.
\newblock {\em Technometrics}, 59(2):262--270.

\bibitem[Fokianos and Pitsillou, 2018]{fokianos2018testing}
Fokianos, K. and Pitsillou, M. (2018).
\newblock Testing independence for multivariate time series via the
  auto-distance correlation matrix.
\newblock {\em Biometrika}, 105(2):337--352.

\bibitem[Gabrys et~al., 2010]{gabrys2010tests}
Gabrys, R., Horv{\'a}th, L., and Kokoszka, P. (2010).
\newblock Tests for error correlation in the functional linear model.
\newblock {\em Journal of the American Statistical Association},
  105(491):1113--1125.

\bibitem[Gabrys and Kokoszka, 2007]{gabrys2007portmanteau}
Gabrys, R. and Kokoszka, P. (2007).
\newblock Portmanteau test of independence for functional observations.
\newblock {\em Journal of the American Statistical Association},
  102(480):1338--1348.

\bibitem[Ghodrati and Panaretos, 2023]{ghodrati2023distributional}
Ghodrati, L. and Panaretos, V.~M. (2023).
\newblock On distributional autoregression and iterated transportation.
\newblock {\em arXiv preprint arXiv:2303.09469}.

\bibitem[Golosnoy et~al., 2012]{golosnoy2012conditional}
Golosnoy, V., Gribisch, B., and Liesenfeld, R. (2012).
\newblock The conditional autoregressive {Wishart} model for multivariate stock
  market volatility.
\newblock {\em Journal of Econometrics}, 167(1):211--223.

\bibitem[Han and Shen, 2021]{han2021generalized}
Han, Q. and Shen, Y. (2021).
\newblock Generalized kernel distance covariance in high dimensions: non-null
  clts and power universality.
\newblock {\em arXiv preprint arXiv:2106.07725}.

\bibitem[Hong, 1996]{hong1996consistent}
Hong, Y. (1996).
\newblock Consistent testing for serial correlation of unknown form.
\newblock {\em Econometrica}, 64(4):837--864.

\bibitem[Hong, 1999]{hong1999hypothesis}
Hong, Y. (1999).
\newblock Hypothesis testing in time series via the empirical characteristic
  function: a generalized spectral density approach.
\newblock {\em Journal of the American Statistical Association},
  94(448):1201--1220.

\bibitem[Horv{\'a}th et~al., 2013]{horvath2013test}
Horv{\'a}th, L., Hu{\v{s}}kov{\'a}, M., and Rice, G. (2013).
\newblock Test of independence for functional data.
\newblock {\em Journal of Multivariate Analysis}, 117:100--119.

\bibitem[Hosking, 1980]{hosking1980multivariate}
Hosking, J.~R. (1980).
\newblock The multivariate portmanteau statistic.
\newblock {\em Journal of the American Statistical Association},
  75(371):602--608.

\bibitem[Janson, 2021]{janson2021asymptotic}
Janson, S. (2021).
\newblock Asymptotic normality for $ m $-dependent and constrained $ u
  $-statistics, with applications to pattern matching in random strings and
  permutations.
\newblock {\em arXiv preprint arXiv:2106.09401}.

\bibitem[Kurtek et~al., 2012]{kurtek2012statistical}
Kurtek, S., Srivastava, A., Klassen, E., and Ding, Z. (2012).
\newblock Statistical modeling of curves using shapes and related features.
\newblock {\em Journal of the American Statistical Association},
  107(499):1152--1165.

\bibitem[Lee, 1990]{lee1990u}
Lee, A.~J. (1990).
\newblock {\em U-statistics: Theory and Practice}.
\newblock Routledge.

\bibitem[Lee et~al., 2020]{lee2020testing}
Lee, C., Zhang, X., and Shao, X. (2020).
\newblock Testing conditional mean independence for functional data.
\newblock {\em Biometrika}, 107(2):331--346.

\bibitem[Leucht and Neumann, 2013]{leucht2013dependent}
Leucht, A. and Neumann, M.~H. (2013).
\newblock Dependent wild bootstrap for degenerate {U}-and {V}-statistics.
\newblock {\em Journal of Multivariate Analysis}, 117:257--280.

\bibitem[Li et~al., 2003]{li2003consistent}
Li, Q., Hsiao, C., and Zinn, J. (2003).
\newblock Consistent specification tests for semiparametric/nonparametric
  models based on series estimation methods.
\newblock {\em Journal of Econometrics}, 112(2):295--325.

\bibitem[Li et~al., 2012]{lzz2013}
Li, R., Zhong, W., and Zhu, L. (2012).
\newblock Feature screening via distance correlation learning.
\newblock {\em Journal of the American Statistical Association},
  107(499):1129--1139.

\bibitem[Li and Mak, 1994]{li1994squared}
Li, W.~K. and Mak, T.~K. (1994).
\newblock On the squared residual autocorrelations in non-linear time series
  with conditional heteroskedasticity.
\newblock {\em Journal of Time Series Analysis}, 15(6):627--636.

\bibitem[Li and McLeod, 1981]{li1981distribution}
Li, W.~K. and McLeod, A.~I. (1981).
\newblock Distribution of the residual autocorrelations in multivariate arma
  time series models.
\newblock {\em Journal of the Royal Statistical Society: Series B},
  43(2):231--239.

\bibitem[Ling and Li, 1997]{ling1997diagnostic}
Ling, S. and Li, W.-K. (1997).
\newblock Diagnostic checking of nonlinear multivariate time series with
  multivariate arch errors.
\newblock {\em Journal of Time Series Analysis}, 18(5):447--464.

\bibitem[Ljung and Box, 1978]{ljung1978measure}
Ljung, G.~M. and Box, G.~E. (1978).
\newblock On a measure of lack of fit in time series models.
\newblock {\em Biometrika}, 65(2):297--303.

\bibitem[Lobato, 2001]{lobato2001testing}
Lobato, I.~N. (2001).
\newblock Testing that a dependent process is uncorrelated.
\newblock {\em Journal of the American Statistical Association},
  96(455):1066--1076.

\bibitem[Lyons, 2013]{lyons2013distance}
Lyons, R. (2013).
\newblock Distance covariance in metric spaces.
\newblock {\em The Annals of Probability}, 41(5):3284--3305.

\bibitem[Lyons, 2014]{lyons2014hyperbolic}
Lyons, R. (2014).
\newblock Hyperbolic space has strong negative type.
\newblock {\em Illinois Journal of Mathematics}, 58(4):1009--1013.

\bibitem[Lyons, 2020]{lyons2020strong}
Lyons, R. (2020).
\newblock Strong negative type in spheres.
\newblock {\em Pacific Journal of Mathematics}, 307(2):383--390.

\bibitem[Panaretos and Tavakoli, 2013]{panaretos2013fourier}
Panaretos, V.~M. and Tavakoli, S. (2013).
\newblock Fourier analysis of stationary time series in function space.
\newblock {\em Annals of Statistics}, 41(2):568--603.

\bibitem[Petersen and M{\"u}ller, 2019]{petersen2019frechet}
Petersen, A. and M{\"u}ller, H.-G. (2019).
\newblock Fr{\'e}chet regression for random objects with euclidean predictors.
\newblock {\em Annals of Statistics}, 47(2):691--719.

\bibitem[Sejdinovic et~al., 2013]{sejdinovic2013equivalence}
Sejdinovic, D., Sriperumbudur, B., Gretton, A., and Fukumizu, K. (2013).
\newblock Equivalence of distance-based and rkhs-based statistics in hypothesis
  testing.
\newblock {\em Annals of Statistics}, 41(5):2263--2291.

\bibitem[Shang, 2017]{shang2017forecasting}
Shang, H.~L. (2017).
\newblock Forecasting intraday {S\&P} 500 index returns: A functional time
  series approach.
\newblock {\em Journal of Forecasting}, 36(7):741--755.

\bibitem[Shao, 2010]{shao2010dependent}
Shao, X. (2010).
\newblock The dependent wild bootstrap.
\newblock {\em Journal of the American Statistical Association},
  105(489):218--235.

\bibitem[Shao, 2011]{shao2011bootstrap}
Shao, X. (2011).
\newblock A bootstrap-assisted spectral test of white noise under unknown
  dependence.
\newblock {\em Journal of Econometrics}, 162(2):213--224.

\bibitem[Sheng and Yin, 2016]{shengyin2016}
Sheng, W. and Yin, X. (2016).
\newblock Sufficient dimension reduction via distance covariance.
\newblock {\em Journal of Computational and Graphical Statistics},
  251(1):91--104.

\bibitem[Steinwart and Christmann, 2008]{steinwart2008support}
Steinwart, I. and Christmann, A. (2008).
\newblock {\em Support vector machines}.
\newblock Springer Science \& Business Media.

\bibitem[Sun, 2005]{sun2005mercer}
Sun, H. (2005).
\newblock Mercer theorem for rkhs on noncompact sets.
\newblock {\em Journal of Complexity}, 21(3):337--349.

\bibitem[Sz{\'e}kely and Rizzo, 2014]{szekely2014partial}
Sz{\'e}kely, G.~J. and Rizzo, M.~L. (2014).
\newblock Partial distance correlation with methods for dissimilarities.
\newblock {\em Annals of Statistics}, 42(6):2382--2412.

\bibitem[Sz{\'e}kely et~al., 2007]{szekely2007measuring}
Sz{\'e}kely, G.~J., Rizzo, M.~L., and Bakirov, N.~K. (2007).
\newblock Measuring and testing dependence by correlation of distances.
\newblock {\em Annals of Statistics}, 35(6):2769--2794.

\bibitem[Wu, 2005]{wuwb2005}
Wu, W.~B. (2005).
\newblock Nonlinear system theory: Another look at dependence.
\newblock {\em Proceedings of the National Academy of Sciences of the United
  States of America}, 102(40):14150--14154.

\bibitem[Yao et~al., 2018]{yzs2018}
Yao, S., Zhang, X., and Shao, X. (2018).
\newblock Testing mutual independence in high dimension via distance
  covariance.
\newblock {\em Journal of Royal Statistical Society: Series B}, 80(3):455--480.

\bibitem[Zhang et~al., 2022]{zhangJTSA2022}
Zhang, C., Kokoszka, P., and Petersen, A. (2022).
\newblock Wasserstein autoregressive models for density time series.
\newblock {\em Journal of Time Series Analysis}, 43:30--52.

\bibitem[Zhang, 2016]{zhang2016white}
Zhang, X. (2016).
\newblock White noise testing and model diagnostic checking for functional time
  series.
\newblock {\em Journal of Econometrics}, 194(1):76--95.

\bibitem[Zhang and Wang, 2016]{zhang2016sparse}
Zhang, X. and Wang, J.-L. (2016).
\newblock From sparse to dense functional data and beyond.
\newblock {\em The Annals of Statistics}, 44(5):2281--2321.

\bibitem[Zhang et~al., 2018]{zhang2018conditional}
Zhang, X., Yao, S., and Shao, X. (2018).
\newblock Conditional mean and quantile dependence testing in high dimension.
\newblock {\em Annals of Statistics}, 46(1):219--246.

\bibitem[Zhou, 2012]{zhou2012measuring}
Zhou, Z. (2012).
\newblock Measuring nonlinear dependence in time-series, a distance correlation
  approach.
\newblock {\em Journal of Time Series Analysis}, 33(3):438--457.

\bibitem[Zhu and M{\"u}ller, 2023a]{zhu2021autoregressive}
Zhu, C. and M{\"u}ller, H.-G. (2023a).
\newblock Autoregressive optimal transport models.
\newblock {\em Journal of the Royal Statistical Society: Series B},
  85(3):1012--1033.

\bibitem[Zhu and M{\"u}ller, 2023b]{zhu2023spherical}
Zhu, C. and M{\"u}ller, H.-G. (2023b).
\newblock Spherical autoregressive models, with application to distributional
  and compositional time series.
\newblock {\em Journal of Econometrics, in press}.

\end{thebibliography}

\end{appendices}
\end{document}